\documentclass[reqno,12pt,a4paper]{amsart}

\usepackage{euscript}
\usepackage{multicol}
\usepackage{amsmath}
\usepackage{times}

\usepackage{amsmath,amsfonts,amsthm,amssymb,amsxtra}
\usepackage{mathrsfs}
\usepackage{amssymb}
\usepackage{amsfonts}
\usepackage{latexsym}
\usepackage{amsthm}
\usepackage[utf8]{inputenc}

\usepackage{mathtools,xparse}

\DeclarePairedDelimiter{\oldnormaux}{\bracevert}{\bracevert}

\NewDocumentCommand{\oldnorm}{som}{%
  \IfBooleanTF{#1}
    {\oldnormaux*{#3}}
    {\IfNoValueTF{#2}
       {\oldnormaux*{\vphantom{dq}#3}}
       {\oldnormaux[#2]{#3}}%
    }%
}
\theoremstyle{plain}
\newtheorem{theorem}{\bf Theorem}[section]
\newtheorem{lemma}[theorem]{\bf Lemma}
\newtheorem{corollary}[theorem]{\bf Corollary}
\newtheorem{proposition}[theorem]{\bf Proposition}

   \def\cH{{\mathcal H}}

\def\o{\omega}

\def\ve{\varepsilon}    \def\vk{\varkappa}

\usepackage{pgfplots}
 
\pgfplotsset{compat = newest}

\numberwithin{equation}{section}

\let\ge\geqslant
\let\le\leqslant
\let\geq\geqslant
\let\leq\leqslant
\newcommand{\ca}{\begin{cases}}
\newcommand{\ac}{\end{cases}}
\newcommand{\ma}{\begin{pmatrix}}
\newcommand{\am}{\end{pmatrix}}
\renewcommand{\[}{\begin{equation}}
\renewcommand{\]}{\end{equation}}
\def\eq{\begin{equation}}
\def\qe{\end{equation}}
\def\[{\begin{equation}}

\title[Anderson localization on    quantum graphs]{ Anderson localization on      quantum graphs coded  by     elements of  a  subshift of finite type}

\author{Oleg Safronov}
\email{osafrono@charlotte.edu}
\address{Department of Mathematics and Statistics,  UNCC,   Charlotte,  NC}

\begin{document}
\maketitle

\thispagestyle{empty}

\begin{abstract}
We study Schrödinger operators on quantum graphs where the number of edges between points is determined by orbits of a "shift of finite type".
We prove Anderson localization for these systems.
\end{abstract}

\section{Main result}

Building on the results of \cite{Saf} regarding the positivity of the Lyapunov exponent, this paper establishes the spectral characteristics and proves localization for quantum graphs with a variable number of edges. 
By intertwining the positivity of the Lyapunov exponent with precise large deviation estimates, we develop a rigorous framework for these disordered systems.
This paper establishes the quantum graph analogue of the celebrated Avila, Damanik, and Zhang result  \cite{ADZ2}.   We prove that for metric graphs governed by a subshift of finite type, the Schrödinger operator exhibits a pure point spectrum on \([0, \infty)\) alongside exponentially decaying eigenfunctions for almost every realization.

For a  positive integer  $\ell>1$,  let ${\mathcal A}^{\Bbb Z}$   be   the collection  of all  infinite sequences $\{\omega_n\}_{n\in {\Bbb Z}}$  such that $\omega_n\in {\mathcal A}$, where ${\mathcal A}= \{1,\dots,\ell\}$.
To introduce a non-trivial topological structure, we restrict our focus to sequences that exclude specific transitions. 
Formally,  given a set of forbidden pairs ${\mathcal F}\subset {\mathcal A}\times {\mathcal A}$,  we consider  the  collection $\Omega\subset{\mathcal A}^{\Bbb Z}$
of all  infinite sequences $\{\omega_n\}_{n\in {\Bbb Z}}$  for which
\[\notag
(\omega_n, \omega_{n+1}) \notin {\mathcal F}, \qquad \forall n\in {\Bbb Z}.
\]
It is easy to see that $\Omega$  is a compact metric space with respect  to
the   metric $d(\cdot,\cdot)$ defined by  $$d(\omega,\omega')=e^{-N(\omega,\omega')},$$
where $N(\omega,\omega')$ is  the  largest  nonnegative integer   such that $\omega_n=\omega_n'$   for all $|n|< N(\omega,\omega')$.
Define the mapping      $T:\,\Omega\to\Omega$     by
$$
\bigl(T\omega\bigr)_n=\omega_{n+1},\qquad  \forall n\in {\Bbb Z}.
$$
Such a mapping $T$ is called a subshift of finite type.   If ${\mathcal F}=\emptyset$, then $T$ is called  the full shift.

For each $\omega\in \Omega$, we construct  the  graph $\Gamma_\omega$,
 displayed  below for the  case  where $\ell=2$, ${\mathcal F}=\emptyset$,  and $\omega=\dots2,1,1,2,2,2,1,1,2,1,2,1,1,1,\dots$

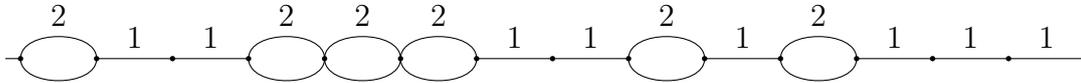
\begin{tikzpicture}
\fill[black] (-3,0) circle (1pt) (-2,0) circle (1pt) (-1,0) circle (1pt) (0,0) circle (1pt) (1,0) circle (1pt) (2,0) circle (1pt)  (3,0) circle (1pt)   (4,0) circle (1pt) (5,0) circle (1pt) (6,0) circle (1pt) (7,0) circle (1pt) (8,0) circle (1pt) 
(9,0) circle (1pt)  (10,0) circle (1pt);
\draw (-3.2,0) --  (-3,0) ;
\draw (-3,0)to[out=90,in=90] node[midway,above] {$2$} (-2,0);
\draw (-3,0)to[out=-90,in=-90]  (-2,0);

\draw (-2,0) -- node[midway,above] {$1$}  (-1,0) ;

\draw (-1,0) --  (0,0) node[midway,above] {$1$} ;

\draw (0,0)to[out=90,in=90] node[midway,above] {$2$} (1,0);
\draw (0,0)to[out=-90,in=-90]  (1,0);
\draw (1,0)to[out=90,in=90] node[midway,above] {$2$} (2,0);
\draw (1,0)to[out=-90,in=-90]  (2,0);
\draw (2,0)to[out=90,in=90] node[midway,above] {$2$} (3,0);
\draw (2,0)to[out=-90,in=-90]  (3,0);
\draw (3,0) -- node[midway,above] {$1$}  (4,0) ;

\draw (4,0) --  (5,0) node[midway,above] {$1$} ;
\draw (5,0)to[out=90,in=90] node[midway,above] {$2$} (6,0);
\draw (5,0)to[out=-90,in=-90]  (6,0);
\draw (6,0) -- node[midway,above] {$1$}  (7,0) ;
\draw (7,0)to[out=90,in=90] node[midway,above] {$2$} (8,0);
\draw (7,0)to[out=-90,in=-90]  (8,0);
\draw (8,0) -- node[midway,above] {$1$}  (9,0) ;
\draw (9,0) -- node[midway,above] {$1$}  (10,0) ;
\draw (10,0) -- node[midway,above] {$1$}  (11,0) ;

\end{tikzpicture}

Namely,
let ${\Bbb Z}$ be the  set of integer numbers.  For each  $\omega\in \Omega$ and  $n\in {\Bbb Z}$,  we  consider  $\omega_n$ copies  of the   interval $[n,n+1]$.
Denoting    these  copies  by $I_{n,j}$, where $j=1,\dots,\omega_n$,
we define the graph $\Gamma_\omega$   as the union
\[
\Gamma_\omega=\bigcup_{n\in {\Bbb Z}}\Bigl(\bigcup_{j=1}^{\omega_n}I_{n,j}\Bigr).
\]
While  the interiors of  the intervals $I_{n,j}$ are assumed to be disjoint,  we will also assume that
their   endpoints are shared in the  sense  that $n$  and $n+1$ belong  to all  intervals  $I_{n,j}$.   Thus,
\[
\bigcap_{j=1}^{\omega_n}I_{n,j}=\{n\}\cup \{n+1\}.
\]

We equip \(\Gamma _{\omega }\) with the natural Lebesgue measure, which restricts to the standard Lebesgue measure on each interval \(I_{n,j}\).  The main object of our study is the Schrödinger operator \(H_{\omega }\), formally defined by
\[
H_{\omega }u=-u^{\prime \prime }.
\]
 The domain \(D(H_\omega)\) of this operator is a subspace of the orthogonal sum of Sobolev spaces on the intervals
\[\notag
D(H_\omega)\subset \bigoplus_{n=-\infty}^\infty\bigoplus_{j=1}^{\omega_n} \quad W^{2,2}(I_{n,j}).
\]
To be in the domain \(D(H_\omega)\), all functions \(u\) must satisfy two requirements:
\begin{enumerate}
\item The functions \(u\) must be continuous on the entire graph \(\Gamma _{\omega }\).
\item Denoting the restriction of \(u\) to \(I_{n,j}\) by \(u_{n,j}\), the derivatives at each vertex \(n \in \mathbb{Z}\) must satisfy:
\begin{equation}\label{Kirchhoff}
\sum_{j=1}^{\omega_n}u'_{n,j}(n) =\sum_{j=1}^{\omega_{n-1}}u'_{n-1,j}(n), \quad \forall n\in {\Bbb Z}.
\end{equation}
\end{enumerate}
The  last  relation  is  called Kirchhoff's gluing   condition  at the point $n$.  Under these boundary and matching conditions, the operator \(H_{\omega }\) is self-adjoint in the Hilbert space \(L^2(\Gamma_\omega)\).

Since \(\Omega \) is a metric space, we can consider probability measures on its Borel \(\sigma \)-algebra. Let \(\mu \) be a \(T\)-ergodic probability measure on \(\Omega \). We restrict our focus to measures satisfying a bounded distortion property. To formally define this, we first introduce the cylinder sets
\[\notag
[n; j_0,\dots, j_k ]=\{\omega\in \Omega:\,\, \omega_{n+s}=j_s,\,\, s=0,\dots ,k\}
\]

\bigskip

{\it Definition}.  The measure  $\mu$  is  said to be a measure   with  a  bounded  distortion property, 
provided there is a  constant $C>1$ such that
\[\notag
C^{-1}\leq \frac{\mu([n;j_0,\dots j_k]\cap [l;i_0,\dots i_s])}{\mu([n;j_0,\dots j_k]) \cdot \mu ( [l;i_0,\dots i_s])}\leq C
\]
for all $l>n+k$  and  $[n;j_0,\dots j_k]\cap [l;i_0,\dots i_s]\neq \emptyset$.

\bigskip
Our main result is the following  theorem.
\begin{theorem}\label{mainthm}
Let $\Omega$ be a subshift of finite type and let $\mu$  be a $T$-ergodic measure with  the bounded  distortion property
such  that ${\rm supp}\, (\mu)=\Omega$. Suppose  $T$ has at least one fixed point and at least one non-fixed  point.   

\begin{enumerate}
\item Then for $\mu$-almost every $\omega$, the operator $H_\omega$
has  pure point spectrum equal to $[0,\infty)$.  
 \item Furthermore,  there is a finite subset ${\frak X}\subset [0,2\pi]$,  independent of $\omega$,  such that
for any   eigenvalue $E$  of $H_\omega$,  if $\sqrt{E}-2\pi k\notin {\frak X}$ for  all $k=0,1,\dots$, then the corresponding eigenfunction
decays exponentially as $|x|\to\infty.$

\end{enumerate}
\end{theorem}

This theorem can  be viewed as an analogue of the result by Avila, Damanik, and Zhang \cite{ADZ2} in which  their framework is changed  
from standard one-dimensional discrete Schrödinger operators to a different setting of quantum graphs while maintaining similar underlying dynamics.
Although we adopt the methodology developed by Avila, Damanik, and Zhang \cite{ADZ2}, certain elements of our proofs differ  to accommodate our specific setting.

\section{Positivity of the Lyapunov exponent. Large  deviations}

The proof of Theorem~\ref{mainthm} hinges on the positivity of the Lyapunov exponent, a property established in \cite{Saf}. Crucially, the conditions on \(\mu \) in \cite{Saf} were even weaker than those required here; it sufficed to assume that \(\mu \) exhibits a local product structure.

To formalize this structural property, we first partition the underlying dynamics into the past and the  future.  Let \(\Omega _{+}\) and \(\Omega _{-}\) denote the spaces of semi-infinite sequences:
$$
\Omega_+=\{\{\omega_n\}_{n\geq0}:\,\, \omega\in \Omega\}
\quad\text{
and 
}\quad
\Omega_-=\{\{\omega_n\}_{n\leq0}:\,\, \omega\in \Omega\}.
$$
Then using  the natural projection  $\pi_\pm$
 from $\Omega$ onto $\Omega_\pm$,  we  define  $\mu_\pm=(\pi_\pm)_*\mu$  on $\Omega_\pm$  to
 be
the pushforward measures of $\mu$.  After  that,  for each $1\leq j\leq \ell$, we introduce
 the
cylinder sets
$$
[0;j]=\{\omega\in\Omega:\,\, \omega_0=j\}\quad
\text{and}\quad
[0;j]_\pm=\{\omega\in\Omega_\pm:\,\, \omega_0=j\}.
$$ 
A local product structure  is  a relation between   the measures $\mu_j=\mu\bigl|_{[0;j]}$  and   the measures $\mu_j^\pm=\mu_\pm\bigl|_{[0;j]}$.
To describe this relation, we need  to consider
the 
natural homeomorphisms
$$P_j: [0;j]\to [0;j]_-\times [0;j]_+   $$ 
defined  by  $$P_j(\omega)=\bigl(\pi_-\omega,\pi_+\omega\bigr),\qquad \forall \omega\in \Omega.$$

\smallskip

{\it Definition}. We say that $\mu$  has a local product structure if there is a positive  density   $\psi:\Omega\to (0,\infty)$
such that for each $1\leq j\leq \ell$, the function  $\psi\circ P_j^{-1}$  belongs to $ L^1\bigl( [0;j]_-\times [0;j]_+, \mu^-_j\times\mu^+_j\bigr)$
and
$$
\bigl(P_j\bigr)_*d\mu_j=\psi\circ P_j^{-1}\, d( \mu^-_j\times\mu^+_j).
$$

\begin{proposition}  Let $k\neq \pi l$  for all $l \in {\Bbb Z}$.  Let $u\in C(\Gamma_\omega)$   be an absolutely  continuous  solution of the  equation
\[\label{u''=k}
-u''(x)=k^2 u(x),\qquad  \text{for  a.e. } \quad  x\in \Gamma_\omega.
\]
satisfying  Kirchhoff's condition \eqref{Kirchhoff} at each $n\in {\Bbb Z}$. Then $u$ satisfies the discrete recurrence relation
\begin{equation}\label{SEquation}
\omega_n u(n+1)+\omega_{n-1}u(n-1) -(\omega_n+\omega_{n-1}) \cos(k) u(n)=0.
\end{equation} 
\end{proposition}

{\it Proof}.
The  solution \(u(x)\) of \eqref{u''=k} on the interval \([n, n+1]\) can be expressed in terms of its boundary values \(u(n)\) and \(u(n+1)\) at the endpoints:\[\notag
u(x)=u(n)\frac{\sin (k(n+1-x))}{\sin (k)}+u(n+1)\frac{\sin (k(x-n))}{\sin (k)},\quad x\in [n,n+1].\]
Taking the derivative of this solution with respect to \(x\),
we obtain \[\notag
u^{\prime }(x)=-ku(n)\frac{\cos (k(n+1-x))}{\sin (k)}+ku(n+1)\frac{\cos (k(x-n))}{\sin (k)}.\]
Therefore, evaluating this at \(x \to n^+\) yields the outgoing right derivative \(u'(n^+)\):\[\notag
u^{\prime }(n^+)=-ku(n)\frac{\cos (k)}{\sin (k)}+ku(n+1)\frac{1}{\sin (k)}.\]
Similarly,  the 
 left derivative \(u'(n^-)\)equals \[\notag
 u^{\prime }(n^-)=-ku(n-1)\frac{1}{\sin (k)}+ku(n)\frac{\cos (k)}{\sin (k)}.\]
Kirchhoff's matching condition at each vertex \(n \in \mathbb{Z}\) weights the outward directed derivatives by the corresponding number of  edges  \(\omega _{n}\) and \(\omega _{n-1}\). The total flux out of vertex \(n\) must equal zero:\[\notag
\omega _{n}u^{\prime }(n^+)-\omega _{n-1}u^{\prime }(n^-)=0.\]
Consequently,
\[\notag
\omega _{n}\Big(-u(n)\cos (k)+u(n+1)\Big)-\omega _{n-1}\Big(-u(n-1)+u(n)\cos (k)\Big)=0.\]
$\,\,\,\Box$

\bigskip
We will say that $E=k^2$ is a generalized eigenvalue for $H_\omega$, if \eqref{u''=k}  has a nontrivial solution obeying
\[\label{genereig}
|u(n)|\leq C_u (1+|n|),\qquad \forall n\in {\Bbb Z}.
\]
A solution satisfying \eqref{genereig} is called a  generalized  eigenfunction.

It  turns  out that spectral properties of $H_\o$ are  closely related  to that  of the operator ${\mathcal H}_\o$ which is defined on $\ell^2({\Bbb Z})$  by 
\[\label{JacobiH}
\bigl[{\mathcal H}_\o u\bigr](n)=\frac{2\o_n u(n+1)}{\sqrt{(\o_{n+1}+\o_n)(\o_n+\o_{n-1})}}+\frac{2\o_{n-1} u(n-1)}{\sqrt{(\o_{n-1}+\o_{n-2})(\o_n+\o_{n-1})}}
\]
for each $u\in \ell^2({\Bbb Z})$.  Generalized  eigenvalues and generalized  eigenfunctions  for ${\mathcal H}_\o$  are introduced in the same way as  for $H_\o$.

\begin{proposition} Let $k\neq \pi l$  for all $l\in {\Bbb Z}$.
The  point $E=k^2>0$ is  a  generalized eigenvalie  of the operator $H_\o$ if and only  if
$2\cos k$ is a  generalized eigenvalie  of the operator ${\mathcal H}_\o$. 
\end{proposition}

{\it Proof.}  It is  enough to note  that  the corresponding  generalized eiegnfunctions
of the two operators ${\mathcal H}_\o$ and $H_\o$  are related to each other  by
\[\notag
v(n)=\sqrt{\omega _{n}+\omega _{n-1}}\,u(n).
\]
Because the weights \(\omega _{n}\) are bounded from above and below by strictly positive constants, the polynomial growth condition on $u(n)$:\[\notag
|u(n)|\le C_{u}(1+|n|)\]
translates directly to the sequence $v(n)$:\[\notag
|v(n)|\le C_{v}(1+|n|).\]
$\,\,\Box$

\bigskip

Spectral properties of $H_\omega$ are related to the 
behavior of  solutions to the equation
\eqref{SEquation}.
On the other hand,   all  solutions
to \eqref{SEquation} can be described in terms of the cocycles  $(T, A^E)$  with  $A^{E}:\, \Omega\to{\rm SL}(2, {\Bbb R})$ defined by
\begin{equation}\label{defineA}
 A^E(\omega)= \sqrt{\frac{\omega_{0}}{\omega_{-1}}} \begin{pmatrix}
\frac{\omega_0+\omega_{-1}}{\omega_0}\cos(k)& -\frac{\omega_{-1}}{\omega_0}\\
1& 0
\end{pmatrix}, \qquad E=k^2.
\end{equation}
Namely,  $u$ is a solution of \eqref{SEquation} if  and only if
$$
\begin{pmatrix} u(n)\\ u(n-1) \end{pmatrix}= \sqrt{\frac{\omega_{-1}}{\omega_{n-1}}}A^E_n(\omega)\cdot  \begin{pmatrix} u(0)\\ u(-1) \end{pmatrix},\qquad \forall n\in {\Bbb Z},
$$
where
$$
A^E_n(\omega)=\begin{cases}A^E(T^{n-1}\omega)\cdots A^E(\omega)\quad \text{if}\quad n\geq 1;\\
[A^E_{-n}(T^n\omega)]^{-1}
\quad \text{if}\quad n\leq -1;\\ 
{\rm Id}\quad \text{if}\quad n=0.
\end{cases}
$$
The Lyapunov exponent  for $A^E$ and $\mu$
is defined by
$$
L(E)=\lim_{n\to\infty}\frac1n \int \ln(\|A^E_n(\omega)\|)d\mu(\omega).
$$
Clearly,  $L(E)\geq 0$.
By Kingman’s subaddive ergodic theorem, 
$$
\frac1n \ln(\|A^E_n(\omega)\|)\qquad \text{converges to}\quad L(E)\qquad \text{as}\quad n\to \infty,
$$
for $\mu$-almost every $\omega\in \Omega$.

The  main result of \cite{Saf} desribes the properties of   the set
\[\notag
{\frak L}(\mu)=\{E\in [0,4\pi^2]:\,\, L(E)=0\}.
\]

\begin{theorem}\label{thm2}Let \(T:\Omega\to\Omega\) be a subshift of finite type and \(\mu \) be a \(T\)-ergodic probability measure with full support (\({\rm supp }\,\mu=\Omega\)) and a local product structure.   Assume \(T\) possesses at least one fixed point and at least one non-fixed point \(\omega\in \Omega\). Then the set \({\frak L}(\mu)\) is finite.  Furthermore, if \(E>4\pi^2\) satisfies \((\sqrt E-2\pi n)^2 \notin {\frak L}(\mu)\) for all \(n=1,2,\dots\), then \(L(E)>0\).
\end{theorem}

To further analyze the dynamics, we introduce the geometric structure of stable and unstable sets.

{\it Definition }.   Let $T:\Omega\to \Omega$ be a subshift of finite type.
The local stable set of a point $\omega\in \Omega$ is defined by
$$
W_{\rm loc}^s(\omega)=\{\omega'\in\Omega:\,\, \omega'_n=\omega_n\quad \text{for}\quad n\geq 0\}
$$
and the local unstable set of $\omega$  is defined by
$$
W_{\rm loc}^u(\omega)=\{\omega'\in\Omega:\,\, \omega'_n=\omega_n\quad \text{for}\quad n\leq 0\}.
$$

\bigskip

For $\omega'\in W_{\rm loc}^s(\omega)$,  we define  the relative transfer matrix $H_{\omega',\omega}^{s,n}$  by
$$
H_{\omega,\omega'}^{s,n}=\bigl[A^E_n(\omega')\bigr]^{-1}A^E_n(\omega).
$$
Since $d(T^j\omega',T^j\omega)\leq e^{-j}$ tends to $0$ as $j\to \infty$,  there is an index $n_0$  for which
$$
H_{\omega,\omega'}^{s,n}=H_{\omega,\omega'}^{s,n_0}\qquad \text{for }\quad  n\geq n_0.
$$
In this case, we define  the stable  holonomy  $H_{\omega,\omega'}^s$ by
$$
H_{\omega,\omega'}^s=H_{\omega,\omega'}^{s,n_0}.
$$
The unstable holonomy $H_{\omega,\omega'}^u$  for $\omega'\in W_{\rm loc}^u(\omega)$  is defined  analogously by
$$
H_{\omega,\omega'}^u=
\bigl[A^E_n(\omega')\bigr]^{-1}A^E_n(\omega)   \qquad \text{for all}\quad n\leq -n_0.
$$

These abstract definitions  of  holonomies work not only for  the cocycle \eqref{defineA},  but also  for  any locally constant  function  $A^E:\Omega\to{\rm SL}(2,{\Bbb R})$.
However, if $A^E$ is   defined by \eqref{defineA}, then the matrices $H_{\omega,\omega'}^s$ and $H_{\omega,\omega'}^u$  become  very specific.

\begin{proposition}
Let $A^E$ be  defined  in \eqref{defineA}.  Then
\[
H_{\omega,\omega'}^{s}=\bigl[A^E(\omega')\bigr]^{-1}A^E(\omega), \qquad  \text{for any}\quad \omega'\in W^s(\omega).
\]
Similarly,
\[
H_{\omega,\omega'}^{u}={\rm Id}, \qquad  \text{for any}\quad \omega'\in W^u(\omega).
\] 
\end{proposition}

The general theory of dynamical systems  tells us that
the  cocycle
$$
(T,A^E): \Omega\times {\Bbb R \Bbb P}^1\to {\Bbb R \Bbb P}^1
$$
defined  by 
$$
(T, A^E)(\omega,\xi)= (T\omega, A^E(\omega) \xi)
$$
has an invariant probability measure $m$  on $\Omega\times {\Bbb R \Bbb P}^1.$ We say that such a measure $m$ projects to $\mu$
if $m(\Delta \times {\Bbb R \Bbb P}^1)=\mu(\Delta)$  for all Borel subsets $\Delta$ of $\Omega$.  Given any $T$-invariant  measure $\mu$  on $\Omega$,
   one can  find a $(T,A^E)$-invariant measure  $m$  that projects to $\mu$ by applying  the standard Krylov-Bogolyubov trick used to construct  invariant measures.

\smallskip

%%%%%%%%%%%%%%%%%%%%%%%%%%%%%%%%%%

{\it Definition}.  Suppose $m$ is  a $(T, A^E)$-invariant probability measure 
on $\Omega\times {\Bbb R \Bbb P}^1$
that projects to $\mu$.   A disintegration of $m$ 
 is a measurable family $\{m_\omega:\quad \omega\in \Omega\}$ of probability measures on ${\Bbb R \Bbb P}^1$
having the property
$$
m(D)=\int_\Omega m_\omega(\{\xi\in {\Bbb R \Bbb P}^1:\,\, (\omega,\xi)\in D\})d\mu(\omega)
$$
for each measurable set  $D\subset \Omega\times{\Bbb R \Bbb P}^1.$

\smallskip
  Existence of  such a disintegration is guaranteed 
by Rokhlin’s theorem. 
Moreover, if
$\{\tilde m_\omega:\quad \omega\in \Omega\}$  is another disintegration of $m$ then  $m_\omega=\tilde m_\omega$  for $\mu$-almost
every $\omega\in \Omega$. 
It is  easy to see  that $m$  is
$(T, A^E)$-invariant if and only if $A^E(\omega)_*m_\omega=m_{T\omega}$
for $\mu$-almost every  $\omega\in \Omega$.

\smallskip

{\it Definition}.  
A $(T,A^E)$-invariant measure $m$  on $\Omega\times {\Bbb R \Bbb P}^1$ that projects to $\mu$ is said to be an s-state for $A^E$ provided it has 
a disintegration $\{m_\omega:\quad \omega\in \Omega\}$  such that   for  $\mu$-almost every $\omega\in \Omega $, 
\begin{enumerate}
\item $$\quad A^E(\omega)_* m_\omega=m_{T \omega},$$
\item  $$\bigl( H^s_{\omega,\omega'}\bigr)_*m_\omega=m_{\omega'}\qquad
\text{ for every}\quad \omega'\in W^s(\omega).$$
\end{enumerate}

\smallskip

{\it Definition}.  
A $(T,A^E)$-invariant measure $m$  on $\Omega\times {\Bbb R \Bbb P}^1$ that projects to $\mu$ is said to be a u-state for $A^E$ provided it has 
a disintegration $\{m_\omega:\quad \omega\in \Omega\}$  such that   for  $\mu$-almost every $\omega\in \Omega $, 
\begin{enumerate}
\item$$\quad A^E(\omega)_* m_\omega=m_{T \omega},$$
\item $$ \bigl( H^u_{\omega,\omega'}\bigr)_*m_\omega=m_{\omega'}\qquad
\text{ for every}\quad \omega'\in W^u(\omega).$$
\end{enumerate}

\bigskip

Let $\mathcal{E}$  be the set of energies $E\geq 0$  for which  there is a measure $m$
that  is both  a u-state and an s-state.   It was  shown in \cite{Saf} that the set $\mathcal{E}$  is discrete.
It is convenient to add  the squares of integer multiples of $\pi$   to $\mathcal{F}$, and consider the union
\[\notag
\tilde {\mathcal E}={\mathcal E} \cup\bigcup_{n=0}^\infty \{ (\pi n)^2 \}.
\]
For $\eta>0$,  we define the set 
\[\notag
B_\eta(\tilde {\mathcal E})=\{ E> 0:\,\,  {\rm dist} (E,\tilde{ \mathcal{E}}) <  \eta \}.
\]
In what follows, $I$  is any compact  subinterval of $[0,\infty)$  that does not intersect $B_\eta(\tilde {\mathcal E})$.

\begin{lemma}
For every $E\in I$ the cocycle $A^E$ has a  unique u-state $m^{u,E}$.
It  depends on $E\in I$ continuously in the weak-$*$ topology.
\end{lemma}

The proof is identical to that of Lemma 3.4 in \cite{ADZ2}.

\smallskip

%%%%%%%%%%%%%%%%%%%%%%%%%%%

Below, we  use the notation $\Omega^\pm_j=\pi^\pm [0; j]$.
For a point  $\omega^{-,j}\in \Omega_j^-$, we define $W^u_{\rm loc}(\omega^{-,j})=W^u_{\rm loc}(\omega)$  where $\omega$ is any point in $\Omega$
with the property $\pi^-\omega=\omega^{-,j}$.

The following lemma was   proved in \cite{ADZ2}.

\begin{lemma}\label{lemma3.7}
For a fixed  $\omega^{-,j}\in \Omega_j^-$, let  $\nu^u$ be  a probability measure
on $W^u_{\rm loc}(\omega^{-,j})$ having the property that
\[\notag
C^{-1}  \leq  \frac{d(\pi_*^+\nu^u)}{d\mu_j^+} \leq C.
\]
Then 
\[\notag
\frac1n\sum_{k=0}^{n-1}T^k_* \nu^u \to  \mu,\quad \text{as}\quad  n\to \infty,
\]  in   the  weak-$*$  topology uniformly in $\omega^{-,j}$, $ 1\leq j\leq \ell$ and $\nu$.
\end{lemma}

Note  that the notion of uniform convergence for this sequence of measures is well-defined because the space of considered probability measures is metrizable.

In what follows,  we use   the notation $F^E$  for the mapping  from $\Omega\times \Bbb{RP}^1$ to itself defined by
\[\notag
F^E(\omega, v)=(T\omega, A^E(\omega) v).
\]

%%%%%%%%%%%%%%%%%%%%%%%%%%%%%%%%%%

\begin{lemma} \label{lemma3.8}
Let $E \in  I$ and $\omega^{-,j}\in \Omega_j^-$.   Suppose $m$  is a probability measure on
$W^u_{loc}(\omega^{-,j})\times {\Bbb R \Bbb P}^1
$,
whose projection $\nu^u$
to $W^u_{
loc}(\omega^{-,j})$
 satisfies the assumptions of
Lemma~\ref{lemma3.7}. 
Then
\[
\frac{1}{n}
\sum_{k=0}^{n-1}
(F^E)
^k_* m\to  m^{u,E}
\]
in the weak-* topology,      uniformly in $\omega^{-,j}$,  $E\in I$,  and such choices of $m$.
\end{lemma}

{\it Proof. }  By compactness of $\Omega\times {\Bbb R \Bbb P}^1$,  we see that 
the   collection of measures of the form
\[
\frac{1}{n}
\sum_{k=0}^{n-1}
(F^E)
^k_* m
\]
has  weak-$*$ accumulation points.
Consider such an accumulation point and denote it by $\tilde m$. 
Without loss of generality, we may just assume that 
\[
\frac{1}{n}
\sum_{k=0}^{n-1}
(F^E)
^k_* m\to  \tilde m.
\] 
We only need to prove that $\tilde m=m^{u,E}$ and that the convergence is uniform.
Clearly,  $\tilde m$ is
invariant under $F^E$  and it projects to the limit  of  the sequence 
\[
\frac1{n}
\sum_{k=0}^{n-1}
T^k_*
\nu^u
\]
in the first component.    By  Lemma~\ref{lemma3.7},  it means that $\tilde m$  projects to $\mu$.

Let us  show now that any disintegration  $\{\tilde m_\omega\}_{\omega\in \Omega}$  of $\tilde m$ is invariant under the
unstable holonomy. 
Since $A$ depends only on the past,  any  unstable  holonomy is  the identity ${\rm Id}$.   Thus we only need to show that 
$\tilde m_\omega =\tilde m_{\omega'}$ for $\omega$ and $\omega'$
in the same local unstable set. 

For this purpose,  we let 
$\tilde m^- = (\pi^-\times {\rm Id})_* \tilde m,$ which is a measure on $\Omega_-\times {\Bbb R \Bbb P}^1$.   

Since $A^E$  depends only  on the past,  it naturally descends to a map on $\Omega_-$.
Namely, $A^E(\omega)=\tilde A^E(\pi^-\omega)$
 for some map $\tilde A^E:\Omega_-\to{\rm SL}(2,{\Bbb R})$, which   we will still denote by  $A^E $. 
 
 Let $F^E_-$ be  the 
action of $(T_-,A_{-1}^E) $ on $\Omega_-\times {\Bbb  RP}^1$  where $T_-$ is the right shift on $\Omega_-$ and 
$A_{-1}^E(\omega_-)=A^E(T_-\omega_-)$.
 Then 
 \[(\pi^-\times {\rm Id}) \circ (T, A^E)^{-1}= (T_-,A_{-1}^E
) \circ (\pi^-\times {\rm Id}),\]
which may be  written in the form $ (\pi_-
\times {\rm  Id}) \circ  (F^E )^{-1} = F^E_-(\pi_-\times {\rm Id}).$
This implies that  $\tilde m^-$ 
is invariant
under $F^E_-$.  Indeed,
\begin{equation*}
\begin{aligned}
(F^E_-)_*\tilde m^- 
&= (F^E_-)_*(\pi^-\times {\rm Id})_*\tilde m \\
&= (\pi^-\times {\rm Id})_*(F^E)^{-1}_*\tilde m \\
&= (\pi^-\times {\rm Id})_*\tilde m = \tilde m^-.
\end{aligned}
\end{equation*}
Let $\{ \tilde m^-_{\omega_- }\} $ be a disintegration of $\tilde m^-$.
By $ F^E_-$-invariance,   we then have
$[A^E(T_-\omega_-)]^{-1}
_*\tilde m^-_{\omega_-}=\tilde m^-_{T_-\omega_-}
$,
or equivalently,
\[ (A^E_{-1}
(\omega_-))_* \tilde m^-_{
\omega_-} = \tilde m^-_{
T_-\omega_-} .
\]
By a special case of \cite{AV}, Lemma 3.4,   the disintegration $\{\tilde m_\omega\}$ of $\tilde m$ can be recovered
from the disintegration
$ \{\tilde m^-_{\omega_-}\}$ of $\tilde m^-$
via
\[
\tilde m_\omega  = \lim_{n\to\infty}
(A^E_{-n}
(\pi_-(T^n\omega)))_*\tilde m^-_{
\pi^-(T^n \omega)}
\]
This  implies that
\[\label{constW}
\tilde m_\omega=\tilde m^-_{\pi^-\omega}
\]
so $\tilde m_\omega$  is  constant on  the local unstable set. 
This concludes the proof that $\tilde m$
is a u-state. Therefore, by uniqueness of the u-state, it must be equal to $m^{u,E}$.
Uniform convergence follows again from uniqueness of the limit.  
Indeed, let $\rho$  be  the metric  on the space of probability measures
defined on $\Omega\times \Bbb{RP}^1$. 
For instance, such a metric may be
chosen as \[
\rho(m, m')=\sum_{s=1}^\infty 2^{-s}\Bigl|\int_{\Omega\times \Bbb{RP}^1} f_s dm-\int_{\Omega\times \Bbb{RP}^1} f_s dm'\Bigr|,
\]
where $\{f_s\}_{s=1}^\infty$ is a dense subset of the unit ball  in $C(\Omega\times \Bbb{RP}^1)$.

%%%%%%%%%%%%%%%%%%%%%%%%%%%%%%%%%%%%%%%%%

Let $\ve>0$.
Suppose there is a sequence of measures $m_j$
satisfying   conditions of the lemma and a sequence of  energies $E_j\to E$ such that
\[\label{bol'shee}
\rho\Bigl(\frac1{n_j}  \sum_{k=0}^{n_j-1} (F^{E_j})^k_* m_j,  \, m^{u,E_j}\Bigr)>\ve \qquad  \forall j.
\]
By compactness, we may assume that $\frac1{n_j}  \sum_{k=0}^{n_j-1} (F^{E_j})^k_* m_j$  converges to some probability measure $\tilde m$.  But then
\[
\frac1{n_j}  \sum_{k=0}^{n_j-1} (F^{E_j})^{k+1}_* m_j-\frac1{n_j}  \sum_{k=0}^{n_j-1} (F^{E_j})^k_* m_j\to 0,\qquad \text{as}\quad j\to \infty
\]
implies that the sequence  $F^{E_j}\tilde m$ also converges to $\tilde m$.  Thus,  the measure $\tilde m$ is $F^E$-invariant.
Repeating the arguments  that lead us to \eqref{constW}, we  conclude  that $\tilde m$ is a u-state  for $A^E$.
On the other  hand,  \eqref{bol'shee} implies that 
\[\notag
\rho\Bigl(  \tilde m, \, m^{u,E} \Bigr)\geq \ve,
\]  which contradicts the fact that $\tilde m=m^{u, E}$.
$\Box$

\medskip

%%%%%%%%%%%%%%%%%%%%%%%%%%%%%%%%%%%%%%%%%%%%%%%

 A disintegration of $\mu$  with respect to  the local unstable sets is a (measurable) family
 of measures $\mu^u_{\omega^-}$ 
 \[\notag
 \{\mu^u_{\omega^-}: \,\,\text{probability  measure on}\, W^u_{\rm loc}(\omega^-)\}_{\omega^-\in \Omega_-}
 \]
 having the property  that
 \[\notag
 \mu(D)=\int_{\Omega_-}\mu^u_{\omega^-}\{ \omega\in D:\,\,\pi^-(\omega)=\omega^-\}  \,d\mu^-(\omega^-)
 \]
 Such a disintegration exists by Rokhlin's theorem.  

For a   continuous  function $\phi$  on $\Omega\times {\Bbb R \Bbb P}^1$, define $S_n(\phi)$
by
\[\notag
S_n(\phi)=\sum_{k=0}^{n-1} \phi\circ (F^E)^k.
\]
\begin{theorem}\label{thm3.6}
For every  $\phi=\bar \phi\in C^\alpha(\Omega\times {\Bbb R \Bbb P}^1)$, $\ve>0$, and $E_0\in I$,
there are positive  constants $C$, $c$, and $r$  such that
\[\notag
\mu^u_{\omega^-}\Bigl\{    \omega\in W^u_{\rm loc}(\omega^-):\,\, \Bigl| \frac1n S_n(\phi)(\omega, v) -\int \phi\, d m^{u, E}\Bigr|>\ve \Bigr\}< C e^{-c n}
\]
uniformly in $(\omega^-, v)\in \Omega_-\times {\Bbb R \Bbb P}^1$  and $E\in I\cap (E_0-r,E_0+r)$.

\end{theorem}

The proof of this theorem relies on the following lemma, which is stated in the remarks following Lemma 3.9 of \cite{ADZ2}. This result is widely known as the Azuma–Hoeffding inequality \cite{Az, Hoe}.

\begin{lemma}\label{lemma3.9}
Let \(\{d_n\}_{n\in \mathbb{N}}\) be a sequence of random variables   such   that
\[\label{21}
{\Bbb E}(d_{n+1}| {\frak F}_{n})=0,
\]  where ${\frak F}_{n}$  is the $\sigma$-algebra  generated  by the functions $d_1,\dots ,d_n$.
Suppose that
\[\notag
\|d_n\|_\infty\leq a,\qquad  \forall n.
\]
Then for every $\ve>0$, there is a constant $c>0$  such that
\[\notag
{\Bbb P}\Bigl( \Bigl| \frac1N \sum_{n=1}^N d_n\Bigr|  >\ve \Bigr)<e^{-\frac{c\ve^2 N}{2a^2}},\qquad  \forall N\geq 1.
\]
\end{lemma}

%%%%%%%%%%%%%%%%%%%%%%%%%%%%%%%%%%%%%

{\it Proof of Theorem ~\ref{thm3.6}.}  While we have made minor cosmetic changes to the proof provided by the authors of \cite{ADZ2},
 the core ideas remain the same.   Given its elegance, we felt it necessary to include this proof.
 
By the bounded distortion property of $\mu$, there exists a  constant $C\geq 1$
such that for each $1\leq j\leq \ell$ and $\mu^-$-almost every $\omega^{-,j}\in \Omega_j^-$
\[\label{25}
 C^{-1}\leq \frac{d( \pi^+_*\mu^u_{\omega^{-,j}})}{d\mu^+_j}\leq C .
\]
Indeed,   let $\omega\in \Omega_j^+$, then
\begin{equation*}
\begin{aligned}
\frac{d( \pi^+_* \mu^u_{\omega^{-,j}})}{d\mu^+_j}(\omega) 
&= \lim_{n\to\infty} \frac{\pi^+_* \mu^u_{\omega^{-,j}}([0,\omega_0,\omega_1,\dots, \omega_n])}{\mu^+_j([0,\omega_0,\omega_1,\dots, \omega_n])} \\
&= \lim_{n\to\infty} \frac{\mu^u_{\omega^{-,j}}(W^u_{\mathrm{loc}}(\omega^{-,j})\cap [0,\omega_0,\omega_1,\dots, \omega_n])}{\mu^+_j([0,\omega_0,\omega_1,\dots, \omega_n])} \\
&= \lim_{n\to\infty}\lim_{l\to\infty} \frac{\int_{[-l;i_{-l},\dots,i_{-1},j]} \mu^u_{\omega^{-,j}}(W^u_{\mathrm{loc}}(\omega^{-,j})\cap [0,\omega_0,\omega_1,\dots, \omega_n]) \, d\omega^{-,j}}{\mu_j^-([-l;i_{-l},\dots,i_{-1}]) \mu^+_j([0,\omega_0,\omega_1,\dots, \omega_n])} \\
&= \lim_{n\to\infty}\lim_{l\to\infty} \frac{\mu([-l;i_{-l},\dots,i_{-1},]\cap[0,\omega_0,\omega_1,\dots, \omega_n])}{\mu_j^-([-l;i_{-l},\dots,i_{-1}]) \mu^+_j([0,\omega_0,\omega_1,\dots, \omega_n])} .
\end{aligned}
\end{equation*}
Without loss of generality, we will  assume that
 inequality \eqref{25} holds not only  almost everywhere,  but for all $\omega^{-,j}\in \Omega^{-,j}$,
and, therefore,  $\mu^u_{\omega^{-,j}}$
satisfies the assumption of
Lemma~\ref{lemma3.7}. 

We will work with  integrals of a  lift  of the measure $\mu^u_{\omega^{-,j}}$ defined as follows.
 Fix any $v\in {\Bbb R \Bbb P}^1$ and set $m$ to be the product of the measures $\mu^u_{\omega^{-,j}}$  and $\delta_v$
 defined on
  $W_{\rm loc}^u(\omega^{-,j})\times {\Bbb R \Bbb P}^1$. Then $m$  satisfies the assumptions of Lemma ~\ref{lemma3.8}.

  Since $\pi^+ T=T_+\pi^+$,  we have the relations $\pi^+ T^s=(T_+)^s\pi^+$,   and $\pi^+_* T^s_*\mu^u_{\omega^{-,j}}=(T_+)^s_*\pi^+_*\mu^u_{\omega^{-,j}}$,   which imply that
  \[\notag
 C^{-1}\leq \frac{d( \pi^+_*T^s_*\mu^u_{\omega^{-,j}})}{d(T_+)^s_*\mu^+_j}\leq C.
\]
By $T_+$-invariance of $\mu^+$,  and   the bounded  distortion property of the measure,  for admissible $j,i_1,\dots,i_s$,  we have
\[\notag
C^{-1}\mu([0;j,i_1,\dots,i_{s-1}])\leq \frac{d((T_+)^s_*\mu^+_j\vert_{[0;j,i_1,\dots, i_s]})}{d\mu_{i_s}^+}\leq C\mu([0;j,i_1,\dots,i_{s-1}]).
\]
The above estimate implies that 
\[\notag
\frac1{\mu^u_{\omega^{-,j}([0;j,i_1,\dots,i_s])}}T^s_*\bigl( \mu^u_{\omega^{-,j}}\vert_{[1;i_1,\dots,,i_s]}\bigr)
\]
is a probability measure
that satisfies \eqref{25}, and thereby,  the assumptions of Lemma~\ref{lemma3.7}. 
  Moreover, the measure 
\[\label{0.17}
\frac1
{m([0 ; j, i_1, . . . , i_s] \times {\Bbb  R\Bbb P}^1)}
(F^E)^s_*\bigl(m \vert_{[1:i_1,...,i_s]\times{\Bbb R \Bbb P}^1}\bigr)
\]
obeys the conditions  of Lemma ~\ref{lemma3.8}.
Now for each  $\omega\in W^u_{\rm loc}(\omega^{-,j})$, and $i\in {\Bbb N}$ we define the set
\[\notag
D_i(\omega)
:=
[1; \omega_1, . . . , \omega_i] \cap W^u_{\rm loc}(\omega^{-,j})\times {\Bbb R \Bbb P}^1.
\]
If $\omega_1=i_1,\dots,\omega_s=i_s$, then  the measure \eqref{0.17}  can be  written as
\[\notag
\frac1
{m(D_s(\omega))}
(F^E)^s_*\bigl(m \vert_{D_s(\omega)}\bigr)
\]
Note also that for every continuous  function $\psi$ on $\Omega$,   we have
\[\notag
\int_{(F^E)^s(D_{i}(\omega))}\psi(\tilde \omega) d(F^E)^s_*m(\tilde \omega)=
\int_{D_{i}(\omega)}\bigl[\psi\circ (F^E)^s\bigr](\tilde \omega) dm(\tilde \omega).
\]
Consider a H\"older continuous
function $\phi\in C^\alpha(\Omega\times {\Bbb R \Bbb P}^1)$. 
 By Lemma ~\ref{lemma3.8} and the facts described above,   given $\varepsilon>0$, 
 there is an $N\geq 1$ such that for every$i\geq 1$ and every $\omega\in W^u_{\rm loc}(w^{-,j})$,
 \[\label{26} \Bigl|
 \frac1{m(D_i(\omega)) } \int_{D_i(\omega)} \frac1{N} S^E_N (\phi\circ(F^E)^i)dm-\int_\Omega \phi \, dm^{u,E}\Bigr|<\varepsilon/4\]
 Define now the functions  $Y_i: W^u_{\rm loc}(\omega^{-,j})\to {\Bbb R}$ by setting
 \[
 Y_i(\omega)=\frac1{m(D_i(\omega))}\int_{D_i(\omega)}S^E_i(\phi)dm.
\]
Obviously, $Y_i$ depends only on $\omega_0, \dots,\omega_i$.  Let ${\frak B}_i$ be the $\sigma$-algebra generated by the functions 
$Y_0,\dots, Y_i$,
which is basically generated by the cylinder sets $[0;n_0,\dots,n_i]$. In particular, the
conditional expectation of  $Y_{i+N}$ with respect to ${\frak B}_i$
is
\begin{equation*}
\begin{split}
\mathbb{E}\Bigl(Y_{i+N} \big\vert \mathfrak{B}_i\Bigr)(\omega) 
&= \frac{1}{m(D_i(\omega))} \sum_{D_{i+N}(\tilde{\omega}) \subset D_i(\omega)} m(D_{i+N}(\tilde{\omega})) Y_{i+N}(\tilde{\omega}) \\
&= \frac{1}{m(D_i(\omega))} \sum_{D_{i+N}(\tilde{\omega}) \subset D_i(\omega)} \int_{D_{i+N}(\tilde{\omega})} S^E_{i+N}(\phi) \, dm \\
&= \frac{1}{m(D_i(\omega))} \int_{D_i(\omega)} S^E_{i+N}(\phi) \, dm.
\end{split}
\end{equation*}
%%%%%%%%%%%%%%%%%%%%%%%%%%%%%%%%%%%%%%%%%%%%%%%%
Clearly,
\[\notag
{\Bbb E}(Y_i\vert {\frak B}_i)=Y_i.
\]
Thus, the estimate \eqref{26} can be rewritten as follows
\[\label{28}
\Bigl|\frac1{N}{\Bbb E}\Bigl(Y_{N+i}-Y_i\vert {\frak B}_i\Bigr)-\int_\Omega \phi dm^{u, E}\Bigr|<\varepsilon/4.
\]
Define now  the sequence 
\[\notag
X_n=Y_{nN}-\sum_{k=1}^n {\Bbb E}\Bigl( Y_{kN}-Y_{(k-1)N}\vert {\frak B}_{(k-1)N}\Bigr).
\]
It is  easy to see that the sequence $\{X_n\}$  is a martingale, that is, \eqref{21} holds for $d_n=X_{n+1}-X_{n}$. Indeed,
\begin{equation}
\label{29}
\begin{split}
X_{n+1}-X_n &= Y_{(n+1)N} - Y_{nN} - \mathbb{E}\Bigl( Y_{(n+1)N}-Y_{nN} \bigm\vert \mathfrak{B}_{nN} \Bigr) \\
&= Y_{(n+1)N} - \mathbb{E}\Bigl( Y_{(n+1)N} \bigm\vert \mathfrak{B}_{nN} \Bigr).
\end{split}
\end{equation}
Since the  $\sigma$-algebra ${\frak F}_n$ generated by  the functions $X_1,\dots,X_n$  is precisely ${\frak B}_{nN}$,
 the above equation implies that for all $n\geq 1$,
 \[
 {\Bbb E}\Bigl(X_{n+1}-X_n\vert\, {\frak F}_n\Bigr)=0
 \]
We claim that 
\[
\sup_n \|X_{n+1}-X_n\|_\infty\leq a<\infty,
\]
  because of the H\"older continuity of  $\phi$.
  Indeed, it is clear that
  $$|X_1|=|Y_N-{\Bbb E}(Y_N-Y_0\vert {\frak B}_0)|\leq (2\|\phi\|_\infty+1) N=CN.$$
On the other  hand, 
$X_{n+1}- X_n $ may be rewritten as
\[\notag
\frac1{m(D_{(n+1)N}(\omega))}\int_{D_{(n+1)N}(\omega)}S_{(n+1)N}^E(\phi)dm-\frac1{m(D_{Nn}(\omega))}\int_{D_{Nn}(\omega)}S_{(n+1)N}^E(\phi)dm.
\]
Clearly, it is  sufficient to prove that the difference of  two values of $S_{(N+1)n}^E(\phi)$  at the points $(\omega',v)$
and $(\omega'',v)$  does not exceed a constant, if $\omega', \omega''\in D_{Nn}(\omega)$.
Note that $A^E_i(\omega')v$
 is independent of $\omega$ for all $\omega'\in D_{nN}(\omega)$
and all  $0\leq i\leq nN$.
So, we may denote $A^E_i(\omega')v$  by $v_i$.
Now, for any $\omega',\omega''\in D_{nN}(\omega)$,
\[\notag
\begin{aligned}
& \vert S_{(n+1)N}^E(\phi)(\omega', v) - S_{(n+1)N}^E(\phi)(\omega'',v)\vert \\& \quad\leq \Bigl(\sum_{i=0}^{nN-1} + \sum_{i=nN}^{(n+1)N-1}\Bigr) \vert\phi((F^E)^i(\omega',v)) - \phi((F^E)^i(\omega'',v))\vert \\
&\quad\leq \sum_{i=0}^{nN-1} \vert\phi(T^i\omega',v_i) - \phi(T^i\omega'',v_i)\vert + CN \quad\leq C\sum_{i=0}^{nN-1} d(T^i\omega',T^i\omega'')^\alpha + CN \\
&\quad\leq C\sum_{i=0}^{nN-1} e^{-\alpha(nN-i)} + CN \quad\leq C_\phi + CN.
\end{aligned}
\]
Thus, it follows from Lemma ~\ref{lemma3.9}  that for every $\delta>0$  and all $n\geq 1$, we have
\[\label{30}
\mu^u_{\omega^{-,j}}  \Bigl\{\omega\in W^u_{\rm loc}(\omega^{-,j}): \, \frac1n X_n>\delta \Bigr\}={\Bbb P}\Bigl\{ \frac1n X_n>\delta \Bigr\}<e^{-\frac{\delta^2}{2a^2}n}
\]
Suppose $\omega\in W^u_{\rm loc}(\omega^{-,j})$  satisfies
\[\label{31}
\Bigl\vert\frac1{nN}S_{nN}^E(\phi)(n,v)-\int_{\Omega\times{\Bbb R \Bbb P}^1} \phi\,dm^{u,E}\Bigr\vert>\varepsilon
\]
 We claim that  this inequality  holds on the set $D_{nN}(\omega)$ with $\varepsilon/2$. 
Indeed, repeating the  arguments that worked for the bound of $\|X_{n+1}-X_n\|_\infty$, we conclude that
for all $\omega'\in D_{nN}(\omega)$ $=[0,\omega_0,\omega_1,\dots,\omega_{nN}]\cap W^u_{\rm loc}(\omega^{-,j})$,
\[\notag
\vert S_{nN}^E(\phi)(\omega', v)-S_{nN}^E(\phi)(\omega,v)\vert<C_\phi.
\]
Combining this estimate with  \eqref{31}, we obtain  that
\[
\Bigl\vert  \frac1{nN}Y_{nN}(\omega)- \int_{\Omega\times{\Bbb R \Bbb P}^1} \phi\,dm^{u,E}  \Bigr\vert>\varepsilon/2
\]
for all $n\geq N_0$, where $N_0$ depends only on $\phi$ and $\varepsilon$.
 Combining the latter  inequality with  \eqref{28},  we conclude that if $\omega$  satisfies \eqref{31}, then
 \begin{align*}
\Bigl\vert \frac{1}{nN} X_{nN} \Bigr\vert
&= \Bigl\vert \frac{1}{nN} Y_{nN} - \frac{1}{n} \sum_{k=1}^n \frac{1}{N} \mathbb{E}\Bigl(Y_{kN} - Y_{(k-1)N} \Big\vert_{\mathfrak{B}_{(k-1)N}} \Bigr) \Bigr\vert \\
&\geq \Bigl\vert \frac{1}{nN} Y_{nN} - \int_{\Omega\times\mathbb{R}\mathbb{P}^1} \phi \, dm^{u,E} \Bigr\vert \\
&\quad - \Biggl\vert \frac{1}{n} \sum_{k=1}^n \biggl( \frac{1}{N} \mathbb{E}\bigl(Y_{kN} - Y_{(k-1)N} \big\vert_{\mathfrak{B}_{(k-1)N}} \bigr) - \int_{\Omega\times\mathbb{R}\mathbb{P}^1} \phi \, dm^{u,E} \biggr) \Biggr\vert > \frac{\varepsilon}{4}.
\end{align*}
 Applying \eqref{30} with $\delta=\varepsilon N/4$, we obtain for all $n\geq N_0$,
 \[\notag
 \mu^u_{\omega^{-,j}}\Bigl\{\omega\in W^u_{\rm loc}(\omega^{-,j}):\,\,
 \Bigl\vert\frac1{nN}S_{nN}^E(\phi)(\omega,v)-\int_{\Omega\times{\Bbb R \Bbb P}^1} \phi\,dm^{u,E}\Bigr\vert>\varepsilon\Bigr\}<e^{-\frac{\varepsilon^2 N}{32a^2 }nN}.
 \]
 It remains to note that the same arguments  could be repeated  for a sequence of indices $n N+l$ replacing the sequence $n N$.   
 The only property of this sequence that we  used  is  that the distance  between two consecutive indices in the sequence  is $N$.
 This completes the proof of Theorem~\ref{thm3.6}. $\,\,\Box$
 
 %%%%%%%%%%%%%%%%%%%%%%%%%%%%%%%%%%%%%%%%%%%%%%%%%%%%%%%%%%%%%%
 
 \bigskip
 
 For $v\in {\Bbb R \Bbb P}^1$, we  define $\underline{v}$ as  the unit vector
 having the same  direction as $v$.
 It is shown in \cite{ADZ2}   that Theorem~\ref{thm3.6}  implies the following result.
 
 \begin{theorem}\label{thm3.5} For any $\ve>0$, there are positive constants $C$ and $c$ for which
 \[\notag
 \mu^u_{\omega^-}\Bigl\{\omega\in W^u_{\rm loc}(\omega^-) :\,\, \bigl| \frac1n \ln \|A^E_n(\omega) \underline{v} \| -L(E)\bigr|                 >\ve  \Bigr\}<Ce^{-cn}
 \]
 uniformly in $(\omega^-,v)\in \Omega_-\times {\Bbb R \Bbb P}^1$,   and $E\in I$.
 \end{theorem}

It  is relatively easy to establish (see \cite{ADZ2}  again)  that Theorem~\ref{thm3.5}  implies Theorem~\ref{thm2.10} below.

 \begin{theorem}\label{thm2.10} For any $\ve>0$, there are positive constants $C$ and $c$ depending only on $\ve$ for which
 \[\notag
 \mu\Bigl\{\omega\in\Omega :\,\, \bigl| \frac1n \ln \|A^E_n(\omega)\| -L(E)\bigr|     >\ve  \Bigr\}<Ce^{-cn}
 \]
 uniformly in  $E\in I$.
 \end{theorem}

We shall refer to this result as the Uniform Large Deviations  (ULD).

 %%%%%%%%%%%%%%
 
 \section{Localization as a consequence of ULD}

It is  convenient for  our purposes  to use the  following lemma  from \cite{ADZ2}.
 \begin{lemma}\label{lemma4.3} Let $F$  be an $\alpha$-H\"older  continuous  function on $\Omega_+$
 obeying the condition $\|F\|_\infty<1$ and $|F(\omega)-F(\omega')|<K d(\omega,\omega')^\alpha$  with $K>1$.
Then for all $\varepsilon>0$ and all $r,n\geq 1$
we have
\[\notag
\mu^+\Bigl\{     \omega^+:\, \Bigl\vert\frac1{r}\sum_{s=0}^{r-1} F(T_+^{ns}\omega)-\int F d\mu^+\Big\vert>\varepsilon\Bigr\}
< \exp\Bigl( -\frac{c\varepsilon^2 n^2r}{\log^2(K\varepsilon^{-1})}\Bigr)
\]
where the constant $c>0$ is independent  of $n, r, K$ and $\varepsilon$.
\end{lemma}

In what follows, we will use  the notations
\[\notag
g_n(\omega, E)=\frac1n\log(\| A_n^E(\omega)\|),
\] and $\Gamma=\sup_{E\in I}\|A^E(\cdot)\|_{\infty}$.
For each $1\leq j\leq  \ell$,  choose $\omega^{(j)}\in [0,j]$
and define $\phi(\omega)=\omega^{(\omega_0)}\wedge \omega$  to be the unique
element that belongs to $W^u_{\rm loc}(\omega^{(\omega_0)})\cap W^s_{\rm loc}(\omega)$.

For each natural number $n$, we define   the function
\[\label{35}
g^+_n(\omega, E)=g_{n-1}(T\omega, E),\qquad \omega\in \Omega.
\]
Then we set
\begin{equation}
    h_n(\omega, E) = g_n(\omega, E) - g^+_n(\omega, E), \qquad \omega \in \Omega.
    \label{34}
\end{equation}
Note that
\[\label{38}
\sup_{E\in I}\|g_n^+(\cdot,E)\|_\infty=\sup_{E\in I}\|g_n(\cdot,E)\|_\infty\leq \ln\Gamma.
\]
By construction,   $g^+_n(\omega, E)$ is constant on $W^s_{\rm loc}(\omega)$ (it depends only on the future).    
Moreover, it depends  only on  $\omega_0,\dots,\omega_{n-1}$. Therefore,  
\[
|g^+_n(\omega, E)-g^+_n(\omega', E)|\leq2e^{n\alpha} \ln \Gamma \cdot  d(\omega,\omega')^\alpha,
\]
for any $\alpha>0$.

\begin{proposition}\label{propohn}
Let $h_n$  be defined  by \eqref{34}, then
\[\label{lng}
|h_n(\omega, E)|\leq \frac{2}{n}\bigl(\ln (\Gamma)+1\bigr).
\]
\end{proposition}

{\it Proof}. Note that
\[\notag
\ln\|A_{n-1}^E(T\omega)\|-\ln\|[A^E(\omega)]^{-1}\|\leq \ln\|A_n^E(\omega)\|\leq \ln\|A_{n-1}^E(T\omega)\|+\ln\|A^E(\omega)\|.
\]
Consequently,
\[\notag
\Bigl|g_n(\omega, E)-\frac{n-1}{n}g_{n-1}(T\omega, E)\Bigr|\leq \frac{1}{n}\bigl(\ln ( \Gamma)+1\bigr).
\]
Finally, since $g_{n-1}(T\omega, E)$  and $g^+_{n}(\omega, E)$ are the same, we obtain \eqref{lng}. $\,\Box$

\bigskip

In the statement below,  we  treat $g_n^+(\cdot,E)$ as a  function defined on $\Omega^+$.

\begin{lemma}\label{lemma4.2}
For every $\varepsilon\in(0,1)$, there is an $n_0=n_0(\varepsilon, \Gamma)$ such that
\[\notag
\mu^+\Bigl\{\omega^+:\,\, \Bigl| \frac1r \sum_{s=0}^{r-1} g^+_n(T_+^{ns}\omega^+, E)-L(E)\Bigr|>\varepsilon  \Bigr\}\leq e^{-c\varepsilon^2 r}
\]
for all $E\in I, r\in {\Bbb N}$, and all $n\geq n_0$. The constant $c>0$ in this inequality  is universal.
\end{lemma}

{\it Proof}.
It is  easy to see that the ULD property of $g_n$ implies a corresponding  property  for $g_n^+$. That is,
for every $\varepsilon>0$, there are  positive constants $C$ and $c$  such that
\[\notag
\mu^+\Bigl\{\omega^+:\,\, \Bigl| g_n^+(\omega^+, E)-L(E)\Bigr|>\varepsilon  \Bigr\}\leq C e^{-cn}
\]
holds  uniformly in $E\in I$. Since the function $g_n^+$ is  bounded by $\ln \Gamma$
and $0<L(E)\leq \ln \Gamma$,  we conclude that  there is an $n_1(\varepsilon)$  such that
\[\notag
\Bigl|\int_{\Omega^+}g_n^+(\omega^+, E)d\mu^+-L(E)\Bigr|\leq \frac{\varepsilon}{10}
\] for all $n>n_1(\varepsilon)$.
We apply Lemma~\ref{lemma4.3} to the  function
$ F=\frac{g_n^+}{2\ln \Gamma}.$ Setting $K=3e^{n\alpha}$, we obtain that
there is an $n_2(\varepsilon, \Gamma)$ such that
\[\notag
\mu^+\Bigl\{\omega^+:\,\, \Bigl| \frac1r \sum_{s=0}^{r-1}g^+_n(T_+^{ns}\omega^+, E)-\int_{\Omega^+} g^+_n(\omega, E)d\mu^+\Bigr|>\varepsilon  \Bigr\}\leq  e^{-c\varepsilon^2 r}
\]
for all $n>n_2(\varepsilon,\Gamma)$.  It remains to combine the  two last  estimates. $\Box$

%%%%%%%%%%%%%%%%%%%%%%%%%%%%%%%

\begin{lemma}\label{lemma4.1}
For every $\varepsilon\in(0,1)$, there is an $n_0=n_0(\varepsilon, \Gamma)$ such that
\[
\mu\Bigl\{ \omega\in \Omega:\,\,  \Bigl| \frac1r\sum_{s=0}^{r-1} g_n(T^{ns+s_0}\omega,E)-L(E)    \Bigr|  >\varepsilon \Bigr\}\leq e^{-c\varepsilon^2 r},
\]
for all $n>n_0$, $s_0\in {\Bbb Z}$ and $r\in {\Bbb N}$.  The constant $c>0$ in this inequality  is universal.
\end{lemma}

{\it Proof}.  Let
\[\notag
B_r^+(\varepsilon)=\Bigl\{ \omega^+\in \Omega^+:\,\,  \Bigl| \frac1r\sum_{s=0}^{r-1} g^+_n(T_+^{ns}\omega^+,E)-L(E)    \Bigr|  >\varepsilon \Bigr\}.
\]
Using Proposition~\ref{propohn}, we  observe  that
\[
|g^+_n(\omega^+,E)-g_n(\omega,E)|\leq \frac{2}{n}\bigl(\ln (\Gamma)+1\bigr).
\]
If $n>4\frac{(1+\ln \Gamma)}{\varepsilon}$, then
\[\notag
  \Bigl| \frac1r\sum_{s=0}^{r-1} g_n(T^{ns}\omega,E)-L(E)    \Bigr|  >\varepsilon \implies \Bigl| \frac1r\sum_{s=0}^{r-1} g^+_n(T_+^{ns}\omega^+,E)-L(E)    \Bigr|  >\varepsilon/2.
\]
Consequently,
\begin{equation*}
\begin{aligned}
\mu\Bigl\{ \omega\in \Omega :\,& \Bigl| \frac{1}{r}\sum_{s=0}^{r-1} g_n(T^{ns}\omega,E) - L(E) \Bigr| > \varepsilon \Bigr\} \\
&\leq \mu\Bigl((\pi^+)^{-1}B_r^+(\varepsilon/2)\Bigr) = \mu^+\Bigl(B_r^+(\varepsilon/2)\Bigr) \leq e^{-c\varepsilon^2 r}.
\end{aligned}
\end{equation*}
Finally,  one can  add $s_0$ to $ns$, because $\mu$  is $T$-invariant. $\qquad \Box$

\bigskip

It is established in the last section of this paper that
there are constants $\beta>0$ and $C>0$ such that
\[\label{32}
|L(E')-L(E)|\leq C|E'-E|^\beta,\qquad \forall E,E'\in I.
\] 

In the proposition below,  instead of  the exponent $2/3$
 obtained  in \cite{ADZ2}, we obtain  a smaller exponent $2/7$, which is justified by our analysis of a special case.
\begin{proposition}\label{propo4.4} 
For any $0<\varepsilon<1$, there is a set $\Omega(\varepsilon)\subset \Omega$ of full $\mu$-measure
such that for any $\omega\in \Omega(\varepsilon)$ there is an index $\tilde n_0=\tilde n_0(
\omega,\varepsilon)$  for which
\[\label{42}
\Bigl|\frac1{n^4}\sum_{s=0}^{n^4-1}g_n(T^{ns+s_0}\omega, E)-L(E)\Bigr|<\varepsilon,
\]
if $n>\max\{\tilde n_0, (\ln(|s_0|+1))^{2/7}\}$  and $E\in I$.
\end{proposition}

{\it Proof}. Let  us  consider the complement  of the set of $\omega$'s  for which \eqref{42}  holds for all $E\in I$:
\[\label{defBns}
B_{n,s_0}=\Bigl\{\omega\in \Omega:\,\,\sup_{E\in I}\Bigl|\frac1{n^4}\sum_{s=0}^{n^4-1}g_n(T^{ns+s_0}\omega, E)-L(E)\Bigr|\geq \varepsilon \Bigr\}.
\]
Let $\vk=|I|$. For any $0<\delta\leq \vk/2$, define the discrete set $I_0=I\cap (\delta {\Bbb Z})$.
Then the cardinality of the set $I_0$ obeys ${\rm card}(I_0)\leq \vk/\delta+1\leq 2\vk/\delta.$
Choose now $\delta=\varepsilon/(3\Gamma^n)$. If necessary,   choose $n$ so large that
\[\label{45}
C\delta^\beta=C\bigl(\frac{\varepsilon}{3\Gamma^n}\bigr)^\beta<\frac{\ve}3
\]
where $C$ and $\beta$ are the same as in \eqref{32}.

Combining \eqref{32} with the inequality
\[
\notag
|g_n(\omega, E)-g_n(\omega, E')|\leq \Gamma^{n-1}|E-E'|,
\]
we obtain 
\[\notag
B_{n,s_0}\subset \bigcup_{E\in I_0}
\Bigl\{\omega\in \Omega:\,\,\Bigl|\frac1{n^4}\sum_{s=0}^{n^4-1}g_n(T^{ns+s_0}\omega, E)-L(E)\Bigr|\geq \frac{\varepsilon }3 \Bigr\}.
\]
Taking $n>n_0(\varepsilon/3,\Gamma)$  and  using the fact that ${\rm card}(I_0)\leq \frac{6\vk \Gamma^n}{\ve}$, we obtain by Lemma~\ref{lemma4.1} that
\[
\mu\bigl(  B_{n,s_0} \bigr)\leq \frac{6\vk \Gamma^n}{\ve}e^{-c\ve^2 n^4}
\]
Using this estimate, we  conclude that  the measures of  the sets
\[\notag
B_n=\bigcup_{|s_0|\leq e^{ n^{7/2}}}B_{n,s_0}
\]
obey the condition
\[\notag
\sum_{n=n_0+1}^\infty \mu(B_n)<\infty.
\]
By the Borel-Cantelli lemma,  almost every point $\omega$ belongs to at most  finitely many of the sets $B_n$.
In other words,  for almost every point $\omega$, there is a number $\tilde n_0(\omega,\ve)$ such that
$\omega\notin B_{n,s_0}$  for all $|s_0|\leq e^{n^{7/2}}$ if $n>\tilde n_0$. It remains to use the definition \eqref{defBns}  of the set $B_{n,s_0}$.
$\,\,\Box$

%%%%%%%%%%%%%%%%%%%%%%%%%%%%%%%%%%%%%%%%%%%%

\bigskip

In what follows below, $E\geq 0$  and $\tilde E\in [-2,2]$  are related   by $\tilde E=2 \cos(\sqrt{E})$.

\bigskip

The obtained  results will be used to estimate the finite  volume Green's function.
Let $P_\Lambda$  be the orthogonal projection from $\ell^2({\Bbb Z})$ onto the set $\ell^2(\Lambda)$  where $\Lambda=\{0,1,\dots,N-1\}$.
Define  the finite rank operator
\[\notag
{\mathcal H}_{\omega, N}=P_\Lambda {\mathcal H}_\omega P_\Lambda
\] on $\ell^2(\Lambda)$.   We are interested in the matrix elements of  its resolvent $G_{\omega, N}^E=({\mathcal H}_{\omega, N}-\tilde E)^{-1}$  which are denoted by
$G_{\omega, N}^E(j,k)$.

Let us denote  the standard basis in $\ell^2({\Bbb Z})$  by $\{e_i\}$ and set $u=G^E_{\omega,N}e_k$.  Assume that $0\leq j\leq k\leq N-1$.
Then
\begin{equation*}
\begin{aligned}
& u(j) \det({\mathcal H}_{\omega,N} - \tilde E) = ({\mathcal H}_{\omega,N} - \tilde E)e_0 \wedge \dots\\ &  \wedge ({\mathcal H}_{\omega,N} - \tilde E)e_{j-1} \wedge ({\mathcal H}_{\omega,N} - \tilde E)u \wedge ({\mathcal H}_{\omega,N} - \tilde E)e_{j+1} \wedge \dots \wedge ({\mathcal H}_{\omega,N} - \tilde E)e_{N-1} \\
&= ({\mathcal H}_{\omega,N} - \tilde E)e_0 \wedge \dots \wedge (-\tilde E){\mathcal H}_{\omega,N}e_{j-1} \wedge e_k \wedge ({\mathcal H}_{\omega,N} - \tilde E)e_{j+1} \wedge \dots \wedge ({\mathcal H}_{\omega,N} - \tilde E)e_{N-1} \\
&= (-1)^{j-k+1} ({\mathcal H}_{\omega,N} - \tilde E)e_0 \wedge \dots \wedge ({\mathcal H}_{\omega,N} - \tilde E)e_{j-1} \wedge ({\mathcal H}_{\omega,N} - \tilde E)e_{j+1} \wedge \dots \\
& \quad \wedge ({\mathcal H}_{\omega,N} - \tilde E)e_{k} \wedge e_k \wedge ({\mathcal H}_{\omega,N} - \tilde E)e_{k+1} \wedge \dots \wedge ({\mathcal H}_{\omega,N} - \tilde E)e_{N-1} \\
&= (-1)^{j-k+1} ({\mathcal H}_{\omega,N} - \tilde E)e_0 \wedge \dots \wedge ({\mathcal H}_{\omega,N} - \tilde E)e_{j-1} \wedge \alpha_{j}e_{j} \wedge \dots \\
& \quad \wedge \alpha_{k-1}e_{k-1} \wedge e_k \wedge ({\mathcal H}_{\omega,N} - \tilde E)e_{k+1} \wedge \dots \wedge ({\mathcal H}_{\omega,N} - \tilde E)e_{N-1},
\end{aligned}
\end{equation*}
 where $\alpha_{k}=\frac{2\omega_{k}}{\sqrt{ (\omega_{k-1}+\omega_{k})(\omega_{k}+\omega_{k+1})}}$.  One should mention that  we  have the agreement that
 $\det({\mathcal H}_{\omega, 0}-\tilde E)=1$.
Consequently,
\[\label{47}
G_{\omega, N}^E(j,k)=\frac{(-1)^{k-j+1}\det({\mathcal H}_{\omega, j}-\tilde E)\det({\mathcal H}_{T^{k+1}\omega, N-k-1}-\tilde E)}{\det({\mathcal H}_{\omega,N}-\tilde E)}\prod_{i=j}^{k-1}\alpha_i,
\] for $0\leq j\leq k\leq N-1$. 
Note that
\[\notag
\prod_{i=j}^{k-1}\alpha_i\leq \frac{2 \sqrt{\omega_j\omega_{k-1}}}{\sqrt{(\omega_{j-1}+\omega_{j})(\omega_{k}+\omega_{k-1})}}\leq 2.
\] 
Another important relation is the  one that expresses $A^E_N(\omega)$ in terms of  determinants  considered  above.
Let $u$  be the formal solution  of the equation $H_\omega u=Eu$ satisfying  the initial conditions $u(-1)=0$ and $u(0)=1$.   
Define $w(n)=\sqrt{\o_n+\o_{n-1}}\, u(n)$.
Then 
\[\notag
({\mathcal H}_{\omega,  N} -\tilde E)w= -\alpha_{N-1}w(N)e_{N-1}\implies \sqrt{\frac{\o_0+\o_{-1}}{\o_N+\o_{N-1}}}=-\alpha_{N-1}u(N) G^E_{\omega, N}(0,N-1).
\]
Therefore,
\[\notag
u(N)=\sqrt{\frac{\o_0+\o_{-1}}{\o_N+\o_{N-1}}} \det(\tilde E-{\mathcal H}_{\omega,N}) \prod_{i=0}^{N-1}\alpha_i^{-1}.
\]
Let now $v$  be the solution  satisfying  the intitial conditions $v(-1)=1$ and $v(0)=0$. Then $-\o_0 v(1)/\o_{-1}=1 $.  
Therefore,
\[\notag
-\frac{\o_0}{\o_{-1}}v(N)=\sqrt{\frac{\o_1+\o_{0}}{\o_N+\o_{N-1}}} \det(\tilde E-{\mathcal H}_{T\omega,N-1}) \prod_{i=1}^{N-1}\alpha_i^{-1}.
\]
The  representations of $u$  and $v$ lead  to
\[\label{48}
A^E_N(\omega)=\sqrt{\frac{\omega_{N-1}}{\omega_{-1}}}\begin{pmatrix}
\frac1{\sqrt{\o_N+\o_{N-1}}} & 0  \\ &\\ 0& \frac1{\sqrt{\o_{N-1}+\o_{N-2}}}
\end{pmatrix} \, \tilde A^E_N(\omega)  \begin{pmatrix}
\sqrt{\o_0+\o_{-1}} & 0  \\ &\\ 0& \sqrt{\o_{-1}+\o_{-2}}
\end{pmatrix}, 
\]
where 
\[\notag
\tilde A^E_N(\omega)=\begin{pmatrix}
\det(\tilde E-{\mathcal H}_{\omega,N}) \prod_{i=0}^{N-1}\alpha_i^{-1}&-\alpha_{-1} \det(\tilde E-{\mathcal H}_{T\omega, N-1})\prod_{i=0}^{N-1}\alpha_i^{-1}\\ \,&\, \\ \det(\tilde E-{\mathcal H}_{\omega,N-1}) \prod_{i=0}^{N-2}\alpha_i^{-1}&-\alpha_{-1} \det(\tilde E-{\mathcal H}_{T\omega, N-2}) \prod_{i=0}^{N-2}\alpha_i^{-1}
\end{pmatrix}.
\]
Combining  \eqref{47} and \eqref{48} , we obtain
\[\notag
|G_{\omega, N}^E(j,k)|\leq C_\ell \frac{\|A^E_j(\omega)\| \|A^E_{N-k}(T^k\omega)\|}{|\det({\mathcal H}_{\omega,N}-\tilde E)|} \prod_{i=0}^{N-1}\alpha_i,
\]  for $0\leq j\leq k\leq N-1$ with  some constant $C_\ell$  that depends only on $\ell$.

\begin{corollary}\label{coln1}
For any $\varepsilon\in (0,1)$ and $\omega\in \Omega(\varepsilon)$  there is an $n_1=n_1(\omega,\varepsilon)$
such that the following is true:
\begin{enumerate}

\item For all $E\in I$,
\[\label{50}
\frac1n
\ln\|A^E_n(T^{s_0}\omega)\|\leq L(E)+2\ve
\]
for all $n\in {\Bbb N}$ and $s_0\in {\Bbb Z}$  satisfying $n\geq \max\{n_1,\ln^{10/7}(|s_0|+1)\}$.

\item  For all $n\in {\Bbb N}$ and $s_0\in {\Bbb Z}$  satisfying $n\geq \ve^{-1} \max\{n_1,8\ln^{10/7}(|s_0|+1)\}$,
\[\label{51}
|G^E_{T^{s_0}\omega, n}(j,k)|\leq  C_\ell \frac{\exp((n-|j-k|)L(E)+C_0\ve n)}{|\det({\mathcal H}_{T^{s_0}\omega,n}-\tilde E)|}
\] provided $\tilde E\notin \sigma({\mathcal H}_{T^{s_0}\omega,n})$ and $j,k\in [0,n-1]\cap{\Bbb Z}$.  The constant $C_0$ in this inequality depends only on $\Gamma$, while the constant $C_\ell$ depends only on $\ell$.

\end{enumerate}
\end{corollary}

%%%%%%%%%%%%%%%%%%%%%%%%%%%%%%%%%%%%%%%%

{\it Proof}.   Set $m=\lceil n^{1/5}\rceil$. Then
\[\notag
A^E_n(T^{s_0}\omega)=[A^E(T^{s_0+n}\omega)]^{-1} \cdots [A^E(T^{s_0+m^5-1}\omega)]^{-1} \prod_{s=0}^{m^4-1}A^E_m(T^{s_0+sm}\omega).
\]
Therefore , since the number of the  factors   containing inverse operators  is $m^5-n\leq 31m^4$, we obtain
\[\label{53}
\|A^E_n(T^{s_0}\omega)\|\leq \Gamma^{31m^4}\prod_{s=0}^{m^4-1}\|A^E_m(T^{s_0+sm}\omega)\|.
\]
Choose $n_1$  so that
\[\label{52}
n_1\geq \max\{ (\tilde n_0(\omega,\varepsilon))^5,18\ve^{-1}, 93^5\}\quad \text{and}\quad \frac{124\ln \Gamma}{n_1^{1/5}-31}<\ve,
\]
where  $\tilde n_0$ is the same as in  Proposition~\ref{propo4.4}.
Now we apply Proposition~\ref{propo4.4} with $m$ replacing $n$ in it.  Observe  that for any $n$  satisfying the conditions of the first part of  the corollary, 
 $m\geq \max\{\tilde n_0(\omega,\ve),\,   \ln^{2/7}(|s_0|+1)\}$.
Therefore,
\begin{equation*}\begin{aligned}
    \frac{1}{n} \ln \|A^E_n(T^{s_0}\omega)\| 
    &\leq \frac{31m^4 \ln\Gamma}{n} + \frac{m^5}{n}(L(E) + \varepsilon) \\
    &\leq \frac{62 \ln\Gamma}{m-31} + \frac{m \varepsilon}{m-31} + L(E) \\
    &\leq L(E) + 2\varepsilon.
\end{aligned}
\end{equation*}
Thus, relation \eqref{50} is established.

To prove \eqref{51}, we   need to estimate the  product
\[\notag
\|A^E_j(T^{s_0}\omega)\|\cdot \|A^E_{n-k}(T^{s_0+k}\omega)\|.
\]
Suppose $n>\ve^{-1}\max\{n_1,8\ln^{10/7}(|s_0|+1)\}$.
Set $h=\lceil \ve n \rceil$. Then for $j\geq0$,  we have
\[\notag
\|A^E_j(T^{s_0}\omega)\|\leq \|A^E_{j+h}(T^{s_0-h}\omega)\|\cdot  \|[A^E_{h}(T^{s_0-h}\omega)]^{-1}\|.
\]
Observe that $j+h\geq h\geq n_1$. Therefore, to apply \eqref{50}, we only need to check  that $h>\ln^{10/7}(|s_0-h|+1)$,  which can be  easily derived  from  
the inequalities
\[\notag\begin{split}
\ln^{10/7}(|a|+|b|+1)\leq 2^{10/7}(\ln^{10/7}(|a|+1)+\ln^{10/7}(|b|+1)),\\
\text{and}\qquad
2^{10/7}\ln^{10/7}(b+1)\leq b,\qquad  \text{for }\quad b>10.
\end{split}
\]
Thus,
\[\label{55}
\|A^E_j(T^{s_0}\omega)\|\leq \exp((j+2h)(L(E)+2\ve))\leq \exp(L(E)j+C_0\ve n),
\]
with the constant $C_0$ depending only on $\Gamma$.

Similarly,
\[\notag
 \|A^E_{n-k}(T^{s_0+k}\omega)\|\leq  \|A^E_{n-k+h}(T^{s_0+k-h}\omega)\|\cdot  \|[A^E_{h}(T^{s_0+k-h}\omega)]^{-1}\|
\]
This time we need to check that $h\geq \ln^{10/7}(|s_0+k-h|+1)$. For that purpose, it is  sufficient to note that
\[\notag
2^{10/7}\ln^{10/7}(|-k+h|+1)\leq 2^{10/7}\ln^{10/7}(n+1)\leq \frac{\ve n}2\leq \frac h 2
\] is implied  by the  condition $n\geq 18 \ve^{-1}$.
Consequently,
\[\label{56}
 \|A^E_{n-k}(T^{s_0+k}\omega)\|\leq \exp((n-k+2h)(L+2\ve))\leq \exp((n-k)L(E) +C_0\ve n)
\]
Multiplying \eqref{55} by \eqref{56}, we obtain \eqref{51} for $j\leq  k$.  Taking the adjoint operator,  we  obtain it for $j\geq k$ as well.  $\,\,\,\,\,\,\Box$

\bigskip

%%%%%%%%%%%%%%%%%%%%%%%%%%%%%%%%%%%%%%%%%%%%%%%%%%%%%%

In \cite{ADZ2},  Avila, Damanik, and Zhang introduce a process termed the "elimination of double resonances." We adopt a specific, context-dependent usage of "resonances," referring to points
$\omega$  rather than energy levels. It can be shown that this set of resonant points in $\Omega$  has small measure. 

For $N\in {\Bbb Z}_+$, we define
\[\notag
\bar N=[N^{\ln N}].
\]
For a  given $\ve$ and $N\in {\Bbb Z}_+$,
the set  $D_N=D_N(\ve)$  is defined  as the set of $\omega\in \Omega$
such that
\[\label{57}
\|G^E_{T^s\omega,[-N_1,N_2]}\|\geq e^{K^2}
\] and
\[\label{58}
\frac1m\ln\| A_m^E(T^{s+r}\omega)\|\leq L(E)-\ve
\]
for some choice of  $s\in {\Bbb Z}$,   $K\geq \max\{ N, \ln^2(|s|+1)\}$,  $E\in I$,  $0\leq N_1,N_2\leq K^9$,
$K^{10}\leq r\leq \bar K$, and $m\in \{K, 2K\}$ (it is either $K$ or $2K$).

\begin{proposition}\label{propo4.6}
For any $\ve>0$ there are positive constants $C,  \eta>0$
for which
\[\notag
\mu(D_N(\ve))\leq C e^{-\eta N},\qquad \forall N\in {\Bbb Z}_+.
\]
\end{proposition}

To prove Proposition~\ref{propo4.6}, we  need the following lemma from \cite{ADZ2}. 
For an admissible $\underline{l}=(l_1,\dots,l_n)$, define $\Omega^+_{\underline{l}}=[0;l]\cap \Omega^+$. Then
we define
\[\label{59}
\mu^+_{\underline{l}}=\frac1{\mu^+(\Omega^+_{\underline{l}})}(T^{|\underline{l}|}_+)_*\bigl(\mu^+\vert_{\Omega^+_{\underline{l}}}\bigr).
\]

\begin{lemma}
Assume that $\mu$ has a bounded distortion property. Then  there is a constant $C\geq 1$
such that
\[\label{61}
\frac{d \mu^+_{\overline{l}}}{d \mu^+}\leq C\qquad \text{a.e. on}\quad \Omega^+,
\]
uniformly in all admissible  $\overline{l}$. In particular,
\[\label{62}
\int_{\Omega_+} f d\mu^+_{\overline{l}}\leq C \int_{\Omega_+} f d\mu^+.
\]
\end{lemma}

\bigskip

{\it Proof of  Proposition~\ref{propo4.6}}.   The proof follows a similar logic to that in \cite{ADZ2}, with the modification that we occasionally replace  the operator 
$H_\o$ by $\cH_\o$ and the point $E$ by  $\tilde E=2\cos \sqrt E$.

Let $N$ and $\ve$ be given.  
Fix $s\in {\Bbb Z}$ and $K\geq \max\{ N, \ln^2(|s|+1)\}$.  
Define the set $D_{K,s}$ as the set of $\omega$  for which \eqref{57},\eqref{58} are satisfied   for some  $N_1$, $N_2$, $r$, and $m$ obeying the conditions
$0\leq N_1,N_2\leq K^9$,
$K^{10}\leq r\leq \bar K$, and $m\in \{K, 2K\}$.
Note that
\[
D_{K,s}=\bigcup_{0\leq N_1,N_2\leq K^9}\bigcup_{K^{10}\leq r\leq \bar K} D_1(N_1,N_2,r,s)\cup  D_2(N_1,N_2,r,s)
\]
where $D_j(N_1,N_2,r,s)$  is the set of $\omega$  for  which there is an $E\in I$  such that  \eqref{57},\eqref{58}  hold  for $m=jK$.

%%%%%%%%%%%%%%%%%%%%%%%%

To estimate  the measure $\mu(D_1)$, let us  assume that $\omega\in D_1(N_1,N_2, r,s)$.  That means \eqref{57} and \eqref{58} hold for some $E\in I$.
Consequently,  there is an $\tilde E_0\in \sigma(\cH_{T^s\omega,[-N_1,N_2]})$  for which
\[
|\tilde E-\tilde E_0|\leq \|G^E_{T^s\omega,[-N_1,N_2]}\|^{-1}\leq e^{-K^2}, \quad \text{with}\quad \tilde E= 2\cos\sqrt E.
\]
Denote also   by $ E_0$ the unique  solution  to  the equation  $\tilde E_0= 2\cos\sqrt {E_0}$ on $I_0$ where 
$I_0$ is  the unique interval of the  form $[(\pi n)^2,(\pi (n+1))^2)$  containing $I$.
Using   the  inequality
\[\notag
|\notag g_n(\omega, E)-g_n(\omega, E_0)|\leq C \Gamma^{n-1}|\tilde E-\tilde E_0|
\]
and selecting $K$ (i.e.  $N$) so large that $C\Gamma^{K}e^{-K^2}<\ve/6$, we  conclude that
\[\notag
g_K(T^{s+r}\omega, E_0)\leq g_K(T^{s+r}\omega, E)+\frac{\ve}6\leq L(E)-\frac{5\ve}6
\]
Applying Hölder  continuity of $L(E)$  stated in \eqref{32} and selecting $K$ so that $Ce^{-\beta K^2}<\ve/6$,
we  derive  the estimate
\[\label{65}
g_K(T^{s+r}\omega, E_0)\leq L(E_0)-\frac{2\ve}3.
\]

Thus,   we see that
\[
D_1(N_1,N_2, r,s)\subset \hat D(N_1,N_2,r,s)
\]
where $\hat D(N_1,N_2,r,s)$ is the set of $\omega$  for which
\[\notag
g_K(T^{s+r}\omega, E_0)\leq L(E_0)-\frac{2\ve}3
\] for  some $E_0\in I_0$  such that  $2 \cos \sqrt{E_0}\in  \sigma(\cH_{T^s\omega,[-N_1,N_2]})$.  
The measure of the set $\hat D$ is the same as  the measure of the set
\[
T^s(\hat D)=\bigcup_{\tilde E_0\in  \sigma(\cH_{\omega,[-N_1,N_2]})}\{\omega:\,\, g_K(T^{r}\omega, E_0)\leq L(E_0)-\frac{2\ve}3\}
\]  where $\tilde E_0=2\cos E_0$.
Let $\underline{l}$  be an admissible  multiindex  such that $|\underline{l}|=K^2+1$. Define $\Omega_{\underline{l}}$ to be the set
 $\Omega_{\underline{l}}= [-K^2; \underline{l}]$.    For each such $\underline l$, we fix $\omega^{(\underline l)}\in \Omega_{\underline{l}}$. Then  all
 points $\omega\in\Omega_{\underline{l}}$  satisfy the condition $d(\omega,\omega^{(\underline l)})\leq e^{-K^2}$. Therefore,   for each $\tilde E_0\in  \sigma(\cH_{\omega,[-N_1,N_2]})$
 there is an  $\tilde E'\in  \sigma(\cH_{\omega^{(\underline l)},[-N_1,N_2]})$  such that
 \[
 |\tilde E_0-\tilde E'_0|\leq \| \cH_{\omega^{(\underline l)},[-N_1,N_2]}-\cH_{\omega,[-N_1,N_2]} \|\leq C e^{-\alpha K^2}
 \]
 Repeating the  arguments  that lead us to \eqref{65}, we obtain  that for $\omega\in\Omega_{\underline{l}}$,
 \[
 g_K(T^r\omega, E_0)<L(E_0)-\frac{2\ve}3\implies  g_K(T^r\omega, E')<L(E')-\frac{\ve}3
 \] for some $\tilde E'\in  \sigma(\cH_{\omega^{(\underline l)},[-N_1,N_2]})$
provided  $K$  satisfies the inequalities $Ce^{-\alpha\beta K^2}<\ve/6$ and  $C\Gamma^{ K}e^{-\alpha K^2}<\ve/6$.
Here, as before, $\tilde E'=2\cos\sqrt{E'}$,  and $E'\in I_0$.

Put  differently, 
\[\label{67}
\hat D \subset  \bigcup_{\underline l}\,\,\bigcup_{\tilde E'\in  \sigma(\cH_{\omega^{ (\underline l) },[-N_1,N_2]} )}\{\omega \in \Omega_{\underline l}:\,\,g_K(T^r\omega, E')<L(E')-\frac{\ve}3\}.
\]
Define  now $S_K(E,\ve)$ to be the set $\{\omega:\,\, g_K(\omega, E)<L(E)-\ve\}$.  Then the set  in the braces on the right hand side  of \eqref{67}
is  the intersection
\[\notag
\Omega_{\underline l}\,\, \bigcap \, T^{-r}[S_K(E',\ve/3)].
\]
To estimate the measure of this set, we  introduce
\[\notag
\tilde S(E)=\bigcup_{\omega\in T^{-K^2}S_K(E,\ve/3)} W_{\rm loc}^s(\omega)
\]
Note that if $\omega'\in W^s_{\rm loc}(\omega)$, then $d(T^{K^2}\omega,T^{K^2}\omega')\leq e^{-K^2}$.
Consequently, $g_K(T^{K^2}\omega, E)=g_K(T^{K^2}\omega', E)$  for $K>1$,  and therefore,
\[\notag
g_K(T^{K^2}\omega', E)\leq L(E)-\frac{\ve}3.
\]
This implies that 
\[\notag
T^{-K^2}S_K(E,\ve/3)=\tilde S(E).
\]
Clearly,
\[
T^{-K^2}\bigl(\Omega_{\underline l}\,\, \bigcap \, T^{-r}[S_K(E',\ve/3)]\bigr)=T^{-K^2}\Omega_{\underline l}\,\, \bigcap \, T^{-r}[\tilde S(E')]
\]
which is  locally $s$-saturated.

%%%%%%%%%%%%%%%%%%%%%%%%%%%%%%%%%%%%%%%

Denote now $S^+(E)=\pi^+ (\tilde S(E))$ and 
observe that $\pi^+\bigl(T^{-K^2}\Omega_{\underline l}\bigr)=[0;\underline l]^+=\Omega_{\underline l}^+.$
All$s$-locally saturated sets $X$ have the following two properties: $\mu(X)=\mu^+(\pi^+(X))$
and $\pi^+T^{-n}X=T_+^{-n} \pi^+ X$. Therefore, since
\[\notag
\pi^+\Bigl[T^{-K^2}\Omega_{\underline l}\,\, \bigcap \, T^{-r}[\tilde S(E')]\Bigr]\subset  \pi^+(T^{-K^2}\Omega_{\underline l})\,\, \bigcap \, \pi^+(T^{-r}[\tilde S(E')]),
\]
we conclude that
\begin{equation*}
\begin{split} 
\mu \Bigl[ T^{-K^2}\Omega_{\underline{l}} \cap T^{-r}\bigl[\tilde{S}(E')\bigr] \Bigr] 
&\leq \mu^+ \Bigl( \pi^+(T^{-K^2}\Omega_{\underline{l}}) \cap \pi^+\bigl(T^{-r}\bigl[\tilde{S}(E')\bigr]\bigr) \Bigr) \\ 
&= \mu^+ \Bigl( \Omega_{\underline{l}}^+ \cap T_+^{-r}\bigl[ S^+(E')\bigr] \Bigr).
\end{split} 
\end{equation*}
Thus, by $T$-invariance of the  measure $\mu$, and  the fact  that $r>|\underline{l}|=K^2+1$,
\[\notag\begin{split}
\mu\bigl(\Omega_{\underline l}\,\, \bigcap \, T^{-r}[S_K(E',\ve/3)]\bigr)\leq \mu^+\Bigl( \Omega_{\underline l}^+\,\, \bigcap \, T_+^{-r}[ S^+(E')]\Bigr)=\\
\mu^+( \Omega_{\underline l}^+)\mu^+_{\underline l}(T_+^{-r+|\underline{l}|}[ S^+(E')])\leq C\mu( \Omega_{\underline l})\mu^+(T_+^{-r+|\underline{l}|}[ S^+(E')])=\\
C\mu( \Omega_{\underline l})\mu^+(S^+(E'))=C\mu( \Omega_{\underline l})\mu(\tilde S(E'))=C\mu( \Omega_{\underline l})\mu(S_K(E',\ve/3)).
\end{split}
\]
Applying the ULD property (on a  slightly bigger interval than $I$), we  obtain
\[\notag
\mu\bigl(\Omega_{\underline l}\,\, \bigcap \, T^{-r}[S_K(E',\ve/3)]\bigr)\leq C_0\, \mu( \Omega_{\underline l})e^{-c\ve K}.
\]
Since the number of the  eigenvalues $\tilde E'\in  \sigma(\cH_{\omega^{ (\underline l) },[-N_1,N_2]} )$  does not exceed $2K^9+1$,
we conclude that
\[\notag
\mu( D_1(N_1,N_2,r,s))\leq 
\mu(\hat D(N_1,N_2,r,s))\leq C(2K^9+1)e^{-c\ve K}\leq C_\ve e^{-c_\ve K}.
\]
Similarly, one  can show that
\[\notag
\mu( D_2(N_1,N_2,r,s))\leq 
 C_\ve e^{-c_\ve K}.
\]
Summing up over all $N_1$, $N_2$ and $r$  obeying the conditions  in  the definition of $D_N(\ve)$,
we  obtain that
\[\notag
\mu(D_{K,s})\leq K^{18}\bar K  C^2_\ve e^{-2c_\ve K}\leq \tilde C e^{-\tilde \eta K}.
\]
Finally, we use the relation
\[\notag
D_N(\ve)= \bigcup_{s\in {\Bbb Z}}  \bigcup_{K\ge \max\{N,  \ln^2(|s|+1)\}} D_{K,s}\subset    \bigcup_{K\ge N} \bigcup_{|s|\leq e^{\sqrt K}}  D_{K,s}
\]
to conclude that
\[\notag
\mu(D_N(\ve))\leq Ce^{-\eta N}.
\]
$\,\,\Box$

\bigskip

%%%%%%%%%%%%%%%%%%%%%%%%%%%%%%%%%%%%%%%%%%%%%%%%%%%%

By Proposition~\ref{propo4.6}, the set
\[
\Omega_+(\ve)=\Omega\setminus \bigcap_{k=1}^\infty\bigcup_{N>k} (D_N(\ve))
\]
is a set of  full measure. The same is true about  the set $\Omega(\ve)$  from Proposition~\ref{propo4.4}.
Since  the set $\Omega_+(\ve)$ is decreasing  in $\ve$  and the set $\Omega(\ve)$ is  increasing in $\ve$,
the set
\[
\notag
\bigl(\bigcap_{\ve\in (0,1)} \Omega_+(\ve)\bigr) \bigcap \bigl( \bigcap_{\ve\in (0,1)}  \Omega(\ve)\bigr)
\] has full  measure.   Let us  consider the intersection
\[
\Omega_*=\Omega_I\bigcap\bigl(\bigcap_{\ve\in (0,1)} \Omega_+(\ve)\bigr) \bigcap \bigl( \bigcap_{\ve\in (0,1)}  \Omega(\ve)\bigr),
\]
where $\Omega_I$ is the set of all $\omega$ for which the spectrum of $H_\omega$   covers the interval $I$.

Let $\ve\in (0,1)$, $\omega\in \Omega_*$  be given and let $u$  be a  generalized eigenfunction of $H_\omega$  corresponding to  $E\in \sigma(H_\omega)$.
We will assume  that $u(s_0)=1$ for some choice of $s_0\in \{0,1\}$.
Define 
\[\label{6.11}
K=\lceil\frac1\ve\max\{N, 2\ln^2(|s_0|+1)\}\rceil
\]
where $N\geq N_0=\max\{\tilde n_0,n_1, n_2\}$  where $\tilde n_0$ and $n_1$ are  taken from  Proposition~\ref{propo4.4} and  Corollary~\ref{coln1}, while $n_2=n_2(\omega,\ve)$
is the number  such that $\omega\notin D_N(\ve)$  for all $N\geq n_2$.  Since $s_0\in \{0,1\}$,  we can replace $\ln(|s_0|+1)$  by $\ln2$ and assume that $K\sim N/\ve$.

\medskip

We consider two claims that hold  for  $\omega\in \Omega_*$ and are similar to the  those in \cite{ADZ2}.

\bigskip

{\bf Claim 1}.  There  are  integer  numbers $a_1$, $a_2$, $b_1$ and $b_2$  having the    properties
\[
-K^9\leq a_1\leq - K^3+1,\qquad 0\leq a_2\leq K^9
\]
and $b_i\in \{a_i+K^3-2, a_i+ K^3-1,  a_i+ K^3\}$  such that
\[
\vert G^E_{T^{s_0}\omega, \Lambda_i}(j,k)  \vert \leq \exp(-|j-k|L(E)+C_0\ve K^3)
\]
for all $j,k\in \Lambda_i=[a_i,b_i)\cap {\Bbb Z}$  where $C_0$ depends only on $\Gamma$.

\bigskip

{\it Proof.}  We apply Proposition~\ref{propo4.4}  twice with:  $s_0$  and $s_0-K^9$.  The  parameter $n$ is  set  equal to $K^3$. As a result, we  obtain
\[
\notag
L(E)-\frac1{K^6}\sum_{s=0}^{K^6-1}g_{K^3}(T^{s_0+sK^3}\omega, E)<\ve
\]
and \[
\notag
L(E)-\frac1{K^6}\sum_{s=0}^{K^6-1}g_{K^3}(T^{s_0+sK^3-K^9}\omega, E)<\ve.
\]
Consequently,  there are  integers $s_1$ and $s_2$ in   the intervals $-K^9\leq s_1\leq -K^3$ and  $0\leq s_2\leq K^9-K^3$
for which
\[\notag
\frac1{K^3}\ln\|A^E_{K^3}(T^{s_0+s_i}\omega)\|>L(E)-\ve.
\]
This leads to   the estimate
\[\notag
\frac1{K^3}(\ln 4+\ln|\det(\cH_{T^{s_0+a_i}\omega,k_i^3}-\tilde E)|)>L(E)-\ve
\]
for $\tilde E=2\cos\sqrt{E}$ and some  choice of 
$k_i\in\{K^3,K^3-1,K^3-2\}$ and $a_i\in\{s_i,s_i+1\}$.  Define  $b_i=a_i+k_i$
and choose $N$  so large that $ \ln4/N_0^3<\ve$.
Then  the preceding inequality will turn into  the estimate
\[\notag
\frac1{K^3}\ln|\det(\cH_{T^{s_0}\omega,\Lambda_i}-\tilde E)|>L(E)-2\ve,
\] where $\Lambda_i=[a_i,b_i)$.
Thus, we  finially obtain that
\[\notag
\begin{aligned}
|G^E_{T^{s_0}\omega,\Lambda_i}(j,k)| 
&\leq \frac{\exp(K^3 - |j-k|L(E) + C_0\ve K^3)}{|\det(\cH_{T^{s_0}\omega,\Lambda_i} - \tilde E)|} \\
&\leq \exp(-|j-k|L(E) + (C_0+2)\ve K^3),
\end{aligned}
\] for $j,k\in \Lambda_i$.  
To apply Corollary~\ref{coln1}, we might  need to enlarge $N_0$   so that $K^3>\frac{4}{\ve }\ln^{10/7}(K^9+2)$
for all $K>N_0$
$\,\,\,\,\Box$

\bigskip

{\bf  Claim 2.}  Let $2\ve C_0< L(E)$.
 Let $u$  be the generalized solution of the equation $H_\omega u=Eu$
satisfying  the conditions $|u(n)|\leq C_u (1+|n|)$  for all $n\in {\Bbb Z}$  and $u(s_0)=1$.
Let $\Lambda_i$ be the same as in Claim 1.
Define 
\[\notag
\ell_i=\Bigl[  \frac{a_i+b_i}2  \Bigr].
\]
Then 
\[\label{claim2}
|u(s_0+\ell_i)|\leq e^{-2K^2},\qquad i=1,2.
\]
whenever $N>N_0$  where $N_0$  is sufficiently large.  Here,  $N_0$ depends on $\omega$, $\ve$ , and $C_u$.

%%%%%%%%%%%%%%%%%%%%%%%%%%%%%%%%%%%%%%

\bigskip

{\it Proof.} Again, we  define $w(n)=\sqrt{\o_n+\o_{n-1}} u(n)$. Then
for any interval $[a,b]$,
\[\label{un=}
w(n)= -G_{\omega,[a,b]}^E(n, a) \alpha_{a-1} w(a-1)-G^E_{\omega,[a,b]}(n,b) \alpha_{b}w(b+1)
\]
for  all $a\leq n\leq b$.  In particular, we may apply this relation   for $n=\ell_i+s_0$ with $\ell_i=[(a_i+b_i)/2]$ and the interval $\Lambda_i+s_0$.
As a result, we obtain
\begin{align*}
|w(\ell_i+s_0)| 
&\leq 2 \bigl( |w(a_i+s_0)|+|w(b_i+s_0-1)| \bigr) \exp(-(K^3/2-1)L(E)+C_0\ve K^3) \\
&\leq C_u (5+|a_i|+|b_i|) \exp(-(K^3/2-1)L(E)+C_0\ve K^3) \\
&\leq C_u (5+2K^9) \exp(-(K^3/2-1)L(E)+C_0\ve K^3).
\end{align*}
Thus, \eqref{claim2}
 holds  for large $K$ and small $\ve$ $\,\,\,\,\Box$
 
 \bigskip

 Now we use the  condition $|u(s_0)|=1$ and apply \eqref{un=} with $a=s_0+\ell_1+1$  and  $b=s_0+\ell_2-1$ to get
 the  estimate
 \[\notag\begin{split}
 1\leq |w(s_0)|\leq 2|G^E_{\omega, [a,b]}(s_0, a)|\,|w(s_0+\ell_1)|+2|G^E_{\omega, [a,b]}(s_0, b)|\,|w(s_0+\ell_2)|\\
 \leq 2\bigl(|G^E_{\omega, [a,b]}(s_0, a)|+|G^E_{\omega, [a,b]}(s_0, b)|\bigr)e^{-2K^2}\leq 4 \|G^E_{\omega, [a,b]}\| e^{-2K^2}.
 \end{split}
 \]
 Consequently,
 \[\notag
  \|G^E_{\omega, [a,b]}\| \geq e^{K^2}.
 \]
 This  condition \eqref{57}  holds, because $0\leq -a,b\leq K^9$.  Since $\omega\notin D_N(\ve)$,  we   conclude that \eqref{58} fails.
 That is,
 \[\label{6.17}
 \frac1m \ln\| A_m^E(T^{s_0+r}\omega) \|\geq L(E)-\ve,
 \]
 for all $N>N_0$,   $K^{10}\leq r\leq \bar K$, and $m\in \{K,2K\}$.

 To proceed, we introduce the following lemma, known as the Avalanche Principle, 
 developed by Goldstein and Schlag \cite{GS} to study the Lyapunov exponents of Schrödinger cocycles.
 \begin{lemma}\label{AvPr}(Avalanche Principle). 
 Let $A^{(1)},\dots,A^{(n)}$ be ${\rm  SL}(2, {\Bbb R})$ matrices satisfying
the following conditions:
\[\notag
\|A^{(j)}\|\geq \lambda>n,\qquad  \forall 1\leq j\leq  n
\]
and
\[\notag
\bigl| \ln\|A^{(j+1)}\|+\ln\|A^{(j)}\|-\ln\|A^{(j+1)}A^{(j)}\|\bigr|< \frac12\ln\lambda, \qquad\forall 1\leq j\leq  n-1.
\]
Then
\[
\Bigl|\ln\|A^{(n)}\dots A^{(1)}\|+\sum_{j=2}^{n-1}\ln\|A^{(j)}\|-\sum_{j=1}^{n-1}\ln\|A^{(j+1)}A^{(j)}\|\leq C\frac n\lambda.
\]
\end{lemma}

%%%%%%%%%%%%%%%%%%%%%%%%%%%%%%%%%%%%%%%
Next, we define
\[
A^{(j)}=A^E_K(T^{s_0+K^{10}+(j-1)K}\omega)
\]
for $1\leq j\leq n$  where $K^{10}\leq n\leq K^{-1}\bar K-K^9$  and $s_0\in\{0,1\}$. Then according to \eqref{6.17},
\[
\|A^{(j)}\|\geq \lambda=\exp(K( L(E)-\ve)).
\]
Moreover, $\lambda>n$ if $N_0$ is  sufficiently large.
Since $\omega\in \Omega(\ve)$ from Corollary~\ref{coln1}, we have
\[\notag
\frac1n\ln\| A_n^E(T^{s}\omega) \|\leq L(E)+2\ve
\]
for all $n>\max\{n_1, \ln^{10/7}(|s|+1)\}$  which implies that 
\[\notag
\|A^{(j)}\|\leq \exp(K(L(E)+2\ve))
\]  as  long as $K>\ln^{10/7}(\bar K+1)$. Thus, \eqref{6.17} with $m=2K$ implies
\[\notag
\begin{aligned}
0 \leq & \ln \|A^{(j+1)}\| + \ln \|A^{(j)}\| - \ln \|A^{(j+1)}A^{(j)}\| \\
\leq & \, 2K(L(E)+2\varepsilon) - 2K(L(E)-2\varepsilon) 
=  \, 8K\varepsilon \\
\leq & \, \frac{1}{2} \ln \lambda = \frac{K(L(E)-\varepsilon)}{2},
\end{aligned}
\]  if $17\ve<L(E)$.

So all conditions of Lemma~\ref{AvPr} are  satisfied.
Denoting $\hat N=n K$  and $r_0=K^{10}$, we   write  the resulting inequality  in the form
\begin{align*}
\ln \| A^E_{\hat N}(T^{s_0+r_0}\omega) \| 
& = \ln \| A^{(n)}\cdots A^{(1)} \| \\
& \geq \sum_{j=1}^{n-1}\ln \| A^{(j+1)} A^{(j)} \| - \sum_{j=2}^{n-1}\ln\| A^{(j)} \| - C \\
& \geq 2(n-1)(L(E)-2\varepsilon)K - (n-2)(L(E)+2\varepsilon)K - C \\
& \geq \hat N (L(E)-5\varepsilon),
\end{align*}
 if $N_0$  is sufficiently large.  Here, $\hat N$  is any number of the form $n K$   between $K^{11}$  and $\bar K-K^{10}$.

Let  us  try   to replace $\hat N$  by an arbitrary sufficiently large number $\ell$.
Observe that  the union  of intervals  $[K^{11}+K^{10},\bar K ]$  contains  all   sufficiently large    integer numbers. Therefore,
we may assume that $\ell\in [K^{11}+K^{10},\bar K ]$ is of the form $\ell=n K+p$  where $0\leq p< K$ and $K^{10}+K^9\leq n\leq K^{-1}\bar K-1$.
In this case,
\begin{align*}
\lVert A^E_{\ell}(T^{s_0}\omega) \rVert 
&\geq \frac{\lVert A^E_{\ell-r_0}(T^{s_0+r_0}\omega) \rVert}{\lVert [A^E_{r_0}(T^{s_0}\omega)]^{-1} \rVert} \geq \frac{\exp\left((\ell-r_0) (L(E)-5\varepsilon)\right)}{\Gamma^{r_0}} \\
&\geq \exp\left(\ell (L(E)-6\varepsilon)\right),
\end{align*}
 for all sufficiently  large $N_0$.  This implies   the   relation
 \[\notag
 \liminf_{n\to\infty}\frac1n
\ln \| A^E_{n}(T^{s_0}\omega) \| \geq	L(E)-6\ve ,
\]
for $s_0\in \{0,1\}$.

\medskip

Let ${\frak G}(H_\omega)$  be the set  of energies $E$  for which the operator $H_\o$
has  a generalized eigenfunction. We have established the following  fact:
\begin{theorem}  There is a subset of  full measure $\Omega_0\subset \Omega$,   such that
for   every $\o\in \Omega_0$, and every $E\in {\frak G}(H_\o)\cap I$,
\[
\lim_{n\to\infty} \frac1n \ln\|  A^E_n(\omega)\|=\lim_{n\to-\infty} \frac1{|n|} \ln\|  A^E_n(\omega)\|=L(E).
\]
\end{theorem}

It was shown in \cite{BDF} that this statement implies Theorem~\ref{mainthm}.

%%%%%%%%%%%%%%%%%%%%%%%%%%%%%%%%

\section{Appendix.  H\"older  continuity of the Lyapunov exponent}

The H\"older  continuity of the Lyapunov exponent \eqref{32} is a key component of our analysis and requires formal justification. 
The  main credit for the results  of this section  belongs  to the  authors of \cite{BDF} as the core methodology was established in their original paper.
We only need to adjust it to  the  specific setting.

Let us  define $L_n(E)$ as the integral
\[\label{a4.5}
L_n(E)=\frac1n \int_\Omega \ln\| A_n^E(\omega)\| \,d\mu
\]
The following lemma, originally proved in \cite{BDF} for a different cocycle, applies directly to our current setting. 

\begin{lemma}
There are constants $c_0,C_0>0$ that depend on $\Gamma$ and the set $I$,
such that
\[
|L(E)+L_n(E)-2L_{2n}(E)|\leq C_0 e^{-c_0n},\qquad \text{for all}\quad E\in I.
\]
\end{lemma}

{\it Proof.} First, choose $\ve>0$  so small
that
\[\notag
0<\frac{4\ve}{\gamma-\ve}<\frac12,\qquad \text{where}\quad \gamma=\sup_{E\in I} L(E).
\]
Without loss of generality, we may assume that the constant $c>0$  in  LDE  satisfies  the condition
\[\notag
c<4(\gamma-\ve).
\]
Since we  only need to prove \eqref{a4.5}  for large  values  of $n$, we may assume that $(\ln n)/c>1$. Then
there is an integer number $N$ such that
\[
e^{cN/5}\leq n\leq e^{cN/4}, \quad \text{which  is smaller than}\quad e^{(L-\ve)N}.
\]
We will also assume that $2nCe^{-c N}\leq e^{-cN/4}$.
 For each $\omega\in \Omega$, consider  the matrices
\[\notag
A^{(j)}(\omega)=A^E_N(T^{(j-1)N}\omega)
\]
According to the ULD,  if  $B(n)$ is the set of $\o$'s on which one of the inequalities 
\[\label{a4.8}
\Bigl|\frac1N \ln \| A^{(j)}(\o)\|-L\Bigr|<\ve,\qquad \Bigl|\frac1N \ln \| A^{(j+1)}(\o) A^{(j)}(\o)\|-L\Bigr|<\ve
\] 
does not hold for some $1\leq j\leq n$,  then
\[
\mu(B(n))\leq e^{-c N/4}.
\]
In particular,   
\[\notag
\|A^{(j)}\|\geq e^{(L-\ve)N}> n
\] and
\[\notag
\Bigl|\ln\|A^{(j+1)}(\o)\|+\ln\|A^{(j)}(\o)\|-\ln\|A^{(j+1)}(\o) A^{(j)}(\o)\|\Bigr|< 4N\ve
\]
for all $\o\notin B(n)$.
Applying the Avalanche Principle  with $\lambda=e^{(L-\ve)N}$,  we obtain  that
\[\notag
\Bigl| \ln \| A^E_{n N} (\o)\|+\sum_{j=2}^{n-1}\ln \|A_N^E(T^{(j-1)N})\|-\sum_{j=1}^{n-1}\ln \|A_{2N}^E(T^{(j-1)N})\|\Bigr|\leq C\frac{n}{\lambda}.
\]
Conseqently,
\[\notag
\Bigl| L_{n N}+\frac{(n-2)}{n}L_N-\frac{2(n-1)}{n}L_{2N}\Bigr|\leq \frac{C}{\lambda N}+4 \mu (B(n)) \ln \Gamma\leq C e^{-c N/4}.
\]
which leads to the inequality
\[\notag
\Bigl| L_{n N}+L_N-2L_{2N}\Bigr| \leq C e^{-c N/4}+\frac4{n}\ln \Gamma\leq C e^{-c N/5}.
\]
%%%%%%%%%%%%%%%%%%%%%%%%%%%%%%%%%%%%%%
Now, consider a sequence  of indices $n_s$  such that  $n_{s+1}$  is an integer multiple of $n_s$, and
\[\notag
e^{c n_s/5}\leq \frac{n_{s+1}}{n_s}\leq \frac12e^{c n_s/4}.
\]
Setting $N=n_s$ and $n$ equal to either $\frac{n_{s+1}}{n_s}$ or $2\frac{n_{s+1}}{n_s}$, we obtain two estimates
\[\label{a4.10}
\Bigl| L_{n_{s+1}}+L_{n_s}-2L_{2n_s}\Bigr|\leq C e^{-c n_{s}/5},
\]
and \[\label{a4.11}
\Bigl| L_{2n_{s+1}}+L_{n_s}-2L_{2n_s}\Bigr|\leq  C e^{-c n_{s}/5}.
\]
Combining these inequalities \eqref{a4.10} and \eqref{a4.11},  we obtain
 \[\label{a4.12}
\Bigl| L_{2n_{s+1}}-L_{n_{s+1}} \Bigr|\leq  2C e^{-c n_{s}/5}.
\]
Now we see that \eqref{a4.10} and \eqref{a4.12} imply
\[
\notag
\Bigl| L_{n_{s+1}}-L_{n_{s}} \Bigr|\leq  \tilde C e^{-c n_{s-1}/5},\qquad  \text{for}\quad s\geq 2.
\]
Therefore,
\[\notag
\Bigl|   L- L_{n_2} \Bigr|=\Bigl|  \sum_{s=2}^\infty ( L_{n_{s+1}}-L_{n_{s}}) \Bigr|\leq \tilde  C e^{-c n_{1}/5}
\]
Thus, writing \eqref{a4.10}  with $s=1$ and replacing $L_{n_2}$ by $L$  in the resulting inequality, we obtain
\[\notag
\Bigl| L+L_{n_1}-2L_{2n_1}\Bigr|\leq C_0 e^{-c n_{1}/5}.
\]
It remains to  observe that the first member $n_1$  of the sequence of indices was an arbitrary sufficiently large number. $\,\,\,\,\Box$

\bigskip

\begin{theorem}
There  is a constant $C>0$ and  an exponent $\beta>0$  such that
\[
|L(E)-L(E')|\leq  C |E-E'|^\beta
\]  for all $E,E'\in I$.
\end{theorem}

{\it Proof}.  Observe that
\[\notag
|L_n(E)-L_n(E')|\leq  C_{\ell, I} \Gamma^{n-1} |E-E'|
\]
Thus,  according to the preceding lemma, we have
\[\notag
|L(E)-L(E')|\leq 3 C_{\ell, I} \Gamma^{2n-1} |E-E'|+C_0 e^{-c_0 n}
\]
for all $n\in {\Bbb N}$.  Put  differently, there is a constant $\tilde C>1$
such that
\[\notag
|L(E)-L(E')|\leq ( \tilde C)^n |E-E'|+\tilde C e^{-c_0 n}
\]
Choosing $n=\Bigl[\frac{\ln |E-E'|^{-1}}{3\ln \tilde C}\Bigr]$, we obtain
\[\notag
|L(E)-L(E')|\leq  |E-E'|^{2/3}+\tilde C |E-E'|^{ \frac{c_0}{3\ln \tilde C}}.
\]
The proof is complete. $\Box$

\end{document}